\documentclass[journal]{IEEEtran}
\IEEEoverridecommandlockouts

\usepackage{amssymb}
\usepackage{enumitem}
\usepackage{amsmath}
\usepackage{amsthm}
\usepackage{amssymb}
\usepackage{graphicx}
\usepackage{epstopdf}
\usepackage{latexsym}
\usepackage{cite}
\usepackage[dvips]{epsfig}
\usepackage{mathrsfs}
\usepackage{dsfont}
\usepackage{fix2col}
\usepackage{algorithm}
\usepackage{algpseudocode}
\usepackage{color,subcaption}

\newcommand{\be}{\begin{eqnarray}}

\newcommand{\ee}{\end{eqnarray}}
\newcommand{\ben}{\begin{eqnarray*}}
\newcommand{\een}{\end{eqnarray*}}
\newcommand{\bfl}{\begin{flalign*}}
\newcommand{\efl}{\end{flalign*}}
\newcommand{\dref}[1]{(\ref{#1})}
\newcommand{\expect}[1]{{\mathbb E} \Bigl[ #1\Bigr]}

\newcommand{\prob}[1]{{\mathbb P} \left( #1\right)}

\newcommand{\calD}{{\mathcal D}}

\newcommand{\calN}{{\mathcal N}}
\newcommand{\calJ}{{\mathcal J}}
\newcommand{\calI}{\mathcal I}

\newcommand{\calS}{{\mathcal S}}
\newcommand{\calT}{{C}}

\newcommand{\calM}{{\mathcal M}}
\newcommand{\calP}{{\mathcal P}}

\newtheorem{theorem}{Theorem}
\newtheorem{definition}{Definition}
\newtheorem{lemma}{Lemma}
\newtheorem{corollary}{Corollary}

\newtheorem{remark}{Remark}

\graphicspath{{Fig/}}

\ifCLASSINFOpdf

\else

\fi

\newcommand{\BNA}{\textsf{BNA}}
\newcommand{\DMA}{\textsf{DMA}}
\newcommand{\DMART}{\textsf{DMA-SRT}}

\newcommand{\OCDMA}{\textsf{DMA-RT}}
\newcommand{\GOCDM}{\textsf{G-DM-RT}}
\newcommand{\GOCDMBF}{\textsf{G-DM-RT-BF}}
\newcommand{\OJOD}{\textsf{G-DM}}
\newcommand{\OJODBF}{\textsf{G-DM-BF}}

\begin{document}

\title{Scheduling Coflows with Dependency Graph}
\author{Mehrnoosh~Shafiee, Javad~Ghaderi
\thanks{The authors are with the Department
of Electrical Engineering, Columbia University, New York,
NY 10027, USA (e-mail: s.mehrnoosh@columbia.edu, jghaderi@ee.columbia.edu). This research was supported by NSF grants CNS-1717867 and CNS-1652115.}}


\maketitle

\begin{abstract}
Applications in data-parallel computing typically consist of multiple stages. In each stage, a set of intermediate parallel data flows (\emph{Coflow}) is produced and transferred between servers to enable starting of next stage. While there has been much research on scheduling isolated coflows, the dependency between coflows in multi-stage jobs has been largely ignored. In this paper, we consider scheduling coflows of multi-stage jobs represented by general \textit{DAG}s (Directed Acyclic Graphs) in a shared data center network, so as to minimize the total weighted completion time of jobs.
This problem is significantly more challenging than the traditional coflow scheduling, as scheduling even a single multi-stage job to minimize its completion time is shown to be NP-hard. 
 In this paper, we propose a polynomial-time algorithm with approximation ratio of $O(\mu\log(m)/\log(\log(m)))$, where $\mu$ is the maximum number of coflows in a job and $m$ is the number of servers.
For the special case that the jobs' underlying dependency graphs are \textit{rooted trees}, we modify the algorithm and improve its approximation ratio.
To verify the performance of our algorithms, we present simulation results using real traffic traces that show up to $53 \%$ improvement over the prior approach. We conclude the paper by providing a result concerning an optimality gap for scheduling coflows with general DAGs.
\end{abstract}

\IEEEpeerreviewmaketitle

\section{Introduction}
\label{CH3introduction}
Modern parallel computing platforms (e.g. Hadoop~\cite{HadoopYarn}, Spark~\cite{zaharia2010spark}, Dryad~\cite{isard2007dryad}) have enabled processing of big data sets in data centers. Processing is typically done through multiple computation and communication stages. While a computation stage involves local operations in servers, a communication stage involves data transfer among the servers in the data center network to enable the next computation stage. 
Such intermediate communication stages can have a significant impact on the application latency~\cite{chowdhury2012coflow}. \textit{Coflow} is an  abstraction that has been proposed to model such communication patterns~\cite{chowdhury2012coflow}. Formally, a coflow is defined as a collection of flows whose completion time is determined by the last flow in the collection. For jobs with a single communication stage, minimizing the average completion times of coflows results in the job's latency improvement. However, for multi-stage jobs, minimizing the average coflow completion time might not be the right metric and might even lead to a worse performance, as it ignores the dependencies between coflows in a job~\cite{tian2019scheduling,tian2018scheduling,chowdhury2015efficient}.

There are two types of dependency between coflows of a multi-stage job: \textit{Starts-After} and \textit{Finishes-Before}~\cite{chowdhury2015efficient}. A Starts-After constraint between two coflows represents an explicit barrier that the second coflow can start only after the first coflow has been completed~\cite{hive}. A Finishes-Before constraint is common when pipelining is used between successive stages~\cite{isard2007dryad}, where 
two dependent coflows can coexist but the second coflow cannot finish until the first coflow finishes. In this paper we focus on scheduling coflows of multi-stage jobs with Starts-After dependency, however, our techniques and results can be easily extended to the other case. Each job is represented by a DAG (\textit{Directed Acyclic Graph}) among its coflows that capture the (Starts-After) dependencies among the coflows. As in~\cite{shafiee2018improved,Chowdhury2014,qiu2015minimizing,tian2018scheduling,tian2019scheduling}, the data center network is modeled as an $m \times m$ switch where $m$ is the number of servers  (see Section~\ref{CH3model} for the formal job and data center network model). As an illustration, Figure~\ref{CH3multistagecoflow} shows one multi-stage job in a $2 \times 2$ switch.    
Given a set of weights, one for each job, our goal is 
to minimize the total weighted completion time of jobs, where the completion time of a job is determined by the completion of the last coflow in its DAG. The weights can capture priorities for different jobs. We state the results as approximation ratios in terms of $m$ (the number of servers), and $\mu$ (the maximum number of coflows in a job).

\begin{figure}[t]
	\centering
	\begin{subfigure}{0.28\columnwidth}
		\includegraphics[width=1 in,height=1.3in]{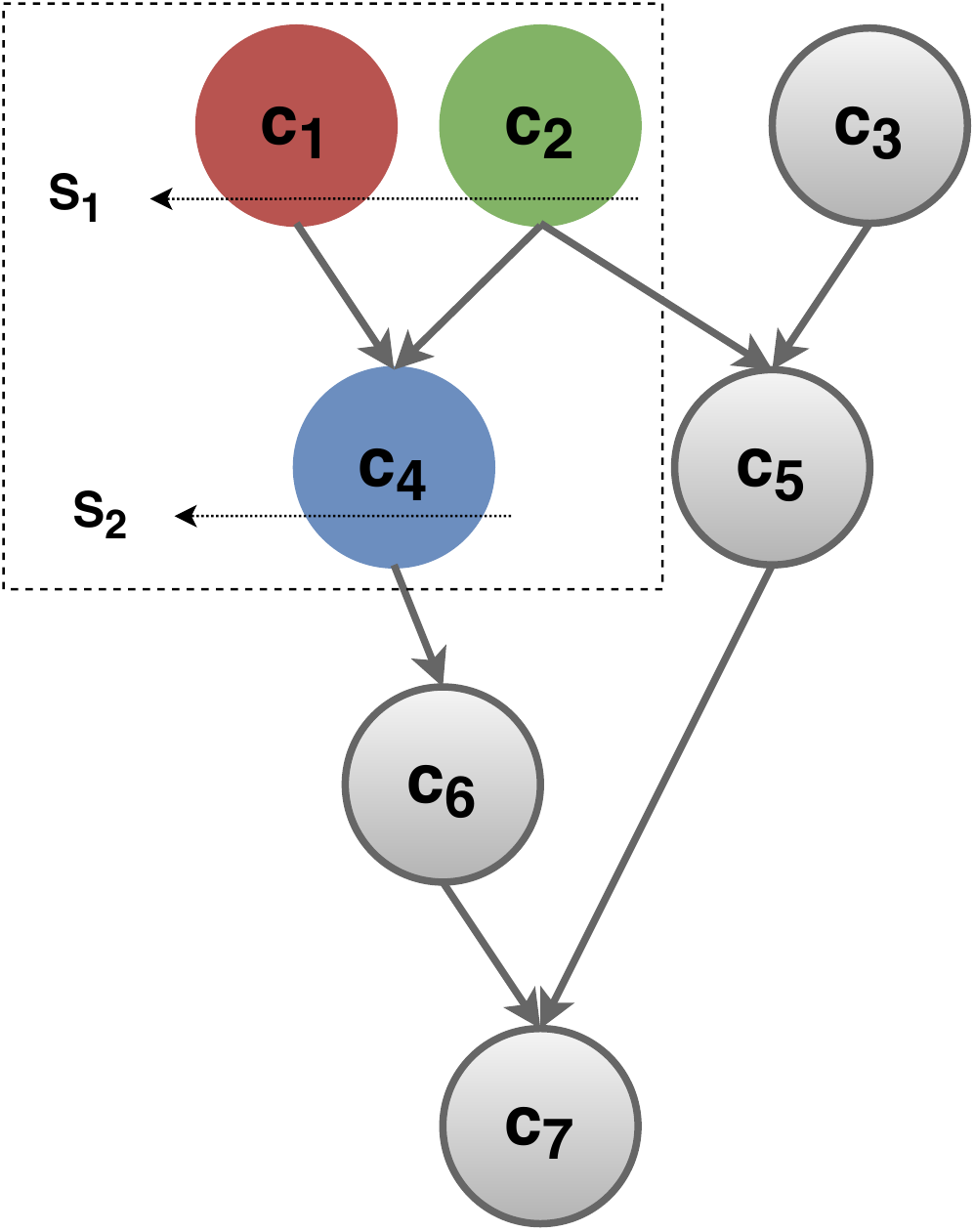}%
		\caption{A multi-stage job with $7$ coflows.}
		\label{CH3jobcoflow}
	\end{subfigure}\hfill
	\begin{subfigure}{0.68\columnwidth}
		\vspace{23pt}
		\includegraphics[width=2.6 in,height=0.8in]{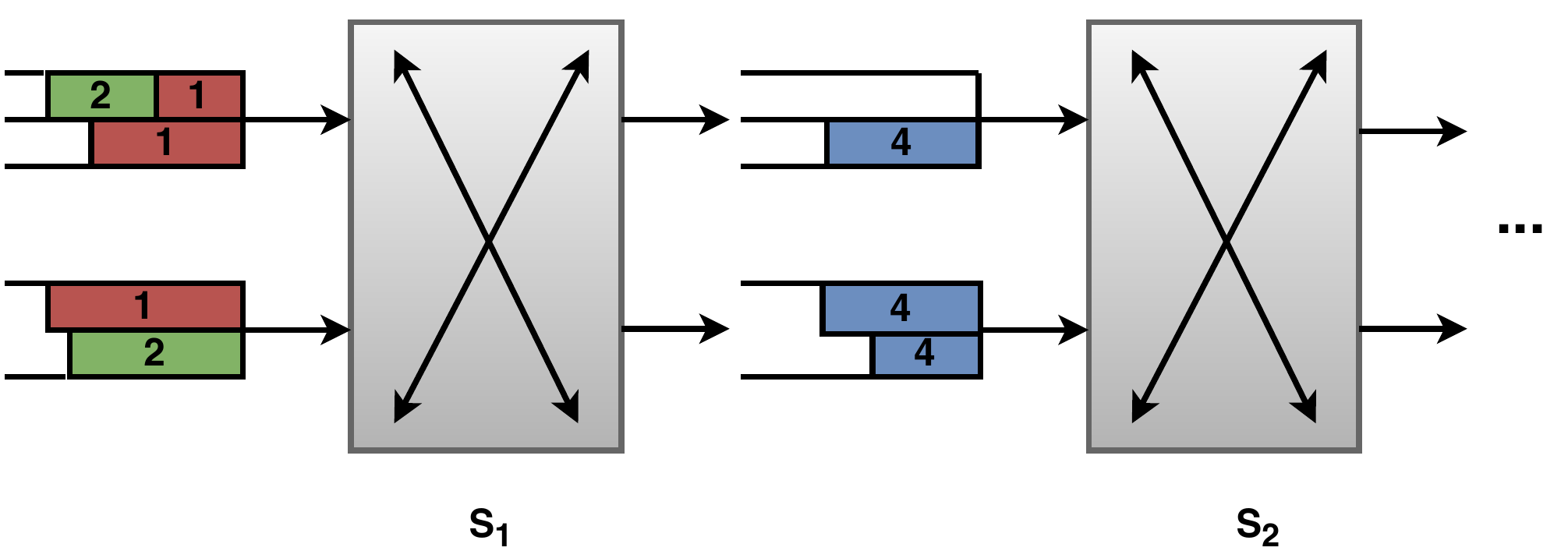}%
		\vspace{14pt}
		\caption{Flows of coflows $1$, $2$, and $4$ and their dependencies in a $2 \times 2$ switch.}
		
		\label{CH3coflowswitch}
	\end{subfigure}\hfill
	\caption{A multi-stage job in a $2 \times 2$ switch. Part of the DAG (in the dashed box) consisting of coflows $1$, $2$, and $4$ is shown in the switch. Coflows $1$ and $2$ can share the network resources at the same time because they are independent (see $\text{S}_1$). Once all their flows are transmitted, flows of coflow $4$ will be ready to be transmitted ($\text{S}_2$ after $\text{S}_1$).}
	\label{CH3multistagecoflow}
\end{figure}

\subsection{Related Work}
The problem considered in this paper can be thought of as a generalization of coflow scheduling that has been widely studied from both theory and system perspectives~\cite{Chowdhury2014,zhao2015rapier,chowdhury2015efficient,qiu2015minimizing,ahmadi2017scheduling,shafiee2018improved,agarwal2018sincronia,Chowdhury2019,im2019matroid,jahanjou2017asymptotically}. 
However, there are only a few works~\cite{tian2019scheduling,tian2018scheduling,chowdhury2015efficient,liu2016scheduling} that consider the multi-stage generalization, with only one algorithm with theoretical performance guarantee~\cite{tian2019scheduling,tian2018scheduling}. Among the heuristics, Aalo~\cite{chowdhury2015efficient} mainly focused on coflow scheduling problem and only provides a brief heuristic to incorporate the multi-stage case. The paper~\cite{liu2016scheduling} proposed a two-level scheduling method based on the most-bottleneck-first heuristic to find the jobs to schedule at each round, and a weighted fair scheduling scheme for intra-job coflow scheduling.
 
The recent papers~\cite{tian2019scheduling,tian2018scheduling} are the most relevant to our work. They consider the problem of scheduling multi-stage job (with Starts-After dependency) to minimize the total weighted job completion times and provide an LP (Linear Program)-based algorithm with $O(m)$ approximation ratio. This algorithm utilizes the technique based on ordering variables, that was also used for coflow scheduling. Their analysis for this algorithm relies on aggregating the load on all the $m$ servers which results in the loss of $O(m)$ in the approximation ratio. \textit{In this paper, we exponentially improve this result by proposing an algorithm that achieves an approximation ratio of $O(\mu g(m))$}, where $\mu$ is the maximum number of coflows in a job, and $g(m)=\log(m)/\log(\log(m))$. Moreover, in the case that the multi-stage job's dependency graph is a rooted tree, we propose an algorithm that achieves an approximation ratio of $O(\sqrt{\mu} g(m) h(m,\mu))$, where $h(m,\mu)=\log(m\mu)/(\log(\log(m\mu))$. We would like to emphasize that the $O(m)$ approximation in~\cite{tian2019scheduling,tian2018scheduling} will not improve if the graph is a rooted tree rather than a general DAG. Note that in practice, the number of coflows in a job is some constant which is much smaller than the number of servers in real-world data centers with hundreds of thousands of servers, i.e., $\mu \ll m$. Also, unlike the $O(m)$ algorithm~\cite{tian2019scheduling,tian2018scheduling}, both of our algorithms are completely combinatorial and do not need to solve a linear program explicitly, hence reducing the complexity. 
A key reason behind the performance improvement in our algorithms is that they utilize the network resources more efficiently by interleaving schedules of coflows of different jobs, unlike the $O(m)$ algorithm~\cite{tian2019scheduling,tian2018scheduling} that schedules coflows one at a time.

Since we represent the dependencies between coflows of a multi-stage job with a \textit{Directed Acyclic Graph} (DAG), DAG scheduling problem is a related line of work. In traditional DAG scheduling, each node represents a task with some processing time and an edge between two nodes indicates the tasks' dependency. There has been extensive results on DAG scheduling problem (DAG-SP) where the goal is to assign tasks to machines in order to minimize the DAG's completion time~\cite{queyranne2006approximation,li2020scheduling,grandl2016graphene,kwok1999static,graham1969bounds,kwok1996dynamic}.

There are also results on DAG-shop scheduling problem (DAG-SSP)~\cite{coffman1976computer,shmoys1994improved,goldberg2001better,schmidt1995chernoff} in which, unlike the DAG-SP, the machine on which each task has to be processed is fixed and no two tasks of the same job can be processed simultaneously.  Our problem of scheduling coflow DAGs is different from the aforementioned problems in several aspects: First, a node in our DAG represent a coflow which itself is a collection of data flows, each with a given pair of source-destination servers. Such couplings are fundamentally different from DAG-SP.  
Second, flows of the same coflow and different unrelated coflows can be scheduled at the same time, which is fundamentally different from DAG-SSP. Hence, algorithms from DAG-SP and DAG-SSP cannot be applied to our problem.

\subsection{Main Contributions} Define $g(m):=\log(m)/\log(\log(m))$, and $h(m, \mu):=\log(m\mu)/\log(\log(m\mu))$. Our main results in this paper can be summarized as follows.

\begin{enumerate}[leftmargin=*] 
	\item[1.] We first prove that even scheduling a multi-stage job to minimize its completion time (makespan) is NP-hard. We then propose an algorithm for minimizing the time to schedule a given set of multi-stage jobs. Our algorithm runs in polynomial time and constructs a schedule in which the makespan is within $O(\mu g(m))$ of the optimal solution for the case that jobs have general DAGs, and $O(\sqrt{\mu} g(m) h(m,\mu))$ when each job is represented as a rooted tree. The algorithms rely on random delaying and merging the greedy schedules of jobs, followed by enforcing the bandwidth constraints. 
	\item[2.] We propose two approximation algorithms for minimizing the total weighted completion time of a given set of multi-stage jobs. For general DAGs, the approximation ratio of our algorithm is $O(\mu g(m))$. For the case of rooted trees, the ratio is improved to $O(\sqrt{\mu} g(m) h(m,\mu))$. Our algorithms are completely combinatorial and do not rely on an explicit solution of a linear program (LP), thus reducing the complexity dramatically. Our approximation algorithms are significant improvements over the LP-based $O(m)$-algorithm of~\cite{tian2019scheduling,tian2018scheduling}.
	\item[3.] To demonstrate the gains in practice, we present extensive simulation results using real traffic traces. The results indicate that our algorithms outperform the $O(m)$-algorithm~\cite{tian2019scheduling,tian2018scheduling} by up to $36 \%$ and $53 \%$ for general DAGs and rooted trees, respectively, in the same settings. 
	\item[4.] We illustrate the existence of instances for which the optimal makespan for a single job with a general DAG is $\Omega(\sqrt{\mu})$ factor larger than two lower bounds for the problem.
	\end{enumerate}
\section{Model and Problem Statement}
\label{CH3model} 
\textit{Network Model:} We consider a cluster of $m$ servers, denoted by the set $\calM$. Each server  has $2$ communication links, one input and one output link with capacity (bandwidth) constraints. For simplicity, we assume all links have equal capacity and without loss of generality, we assume that all the link capacities are normalized to one. Similar to the models in~\cite{shafiee2018improved,Chowdhury2014,qiu2015minimizing,tian2018scheduling}, we abstract out the data center network as one giant non-blocking switch. Each server in the set $\calM$ is represented by one sender server and one receiver server. Therefore, we have an $m \times m$ switch, where the $m$ sender (source) servers on one side, denoted by set $\calM_S$, connected to $m$ receiver (destination) servers on the other side, denoted by set $\calM_R$.

\textit{Job Model:} There is a collection of $n$ multi-stage jobs, denoted by the set $\calN$. 
Each job $j\in \calN$ consists of $\mu_j$ coflows that need to be processed in a given (partial) order. 
Each coflow $c$ of job $j$ is a collection of flows denoted by an $m \times m$ demand matrix $\calD^{(cj)}$. Every flow is a quadruple $(s,r,c,j)$, where $s \in \calM_S$ is its source server, $r \in \calM_R$ is its destination server, and $c$ and $j$ are the coflow and the job to which it belongs. The size of flow $(s,r,c,j)$, denoted by $d_{sr}^{cj}$, is the $(s,r)$-th element of the matrix $\calD^{(cj)}$. For two coflows $c_1, c_2 \in j$, we say coflow $c_1$ \textit{precedes} coflow $c_2$, and denote it by $c_1 \prec c_2$, if all flows of $\calD^{(c_1j)}$ should complete before we can start scheduling any flow of $\calD^{(c_2j)}$ (i.e., Starts-After dependency). We use a DAG $G_j$ to represent the dependency (partial ordering) among the coflows in job $j$, i.e., nodes in $G_j$ represent the coflows of job $j$ and directed edges represent the dependency (precedence constraint) between them. We use $\mu=\max_{j \in \calN} \mu_j$ to denote the maximum number of coflows in any job. 
Figure~\ref{CH3multistagecoflow} illustrates a multi-stage job in a $2 \times 2$ switch network.

\textit{Scheduling Constraints:} Without loss of generality, we assume file sizes of flows are integers and the smallest file size is at least one which is referred to as a \emph{packet}. Scheduling decisions are restricted to such data units (packets), i.e., each sender server can send at most one packet in every time unit (time slot) and each receiver server can receive at most one packet in every time slot, and the feasible schedule at any time slot has to form a \textit{matching} of the switch's bipartite graph. Note that the links' capacity constraints are captured by matching constraints, similarly to models in~\cite{tian2019scheduling,shafiee2018improved,qiu2015minimizing,tian2018scheduling}. 
Further, in a valid schedule, all the precedence constraints in any DAG $G_i$ have to be respected.

\textit{Optimization Objective:} A job is called completed only when all of its coflows finish their processing. Define $C_{cj}$ to be the completion time of coflow $(c,j)$. Then, the completion time of job $j$, denoted by $C_j$, is equal to completion time of its last coflow, i.e., $C_j=\max_{c \in j} C_{cj}$.
The total time that it takes to complete all the jobs in the set $\calN$ is called \textit{makespan} which we denote it by $\calT^{(\calN)}$. Note that by definition $\calT^{(\calN)}=\max_{j \in \calN} C_j$. Given a set of jobs, our first objective is to minimize $\calT^{(\calN)}$. Next, given positive weights $w_j$, $j \in \calN$, we consider the problem of minimizing \textit{the sum of weighted job completion times} defined by $\sum_{j \in \calN} w_j C_j$. The weights can capture different priority for different jobs. 
In the special case that all the weights are equal, the problem is equivalent to minimizing the average job completion time.

\section{Definitions and Preliminaries}
We first present a few definitions and preliminaries regarding complexity of the scheduling problem, and how to optimally schedule a single job whose graph is a path using known results.
\subsection{Definitions}
\begin{definition}[Server Load and Effective Size of a Coflow] \label{CH3def1}
	Suppose a coflow $\calD=\big (d_{sr} \big )_{s,r=1}^m$ is given. Define
	\be
	\label{CH3restrictedsize}
	d_s= \sum_{r \in \calM_R} d_{sr};\ \ d_{r}= \sum_{s \in \calM_S} d_{sr},
	\ee
then $d_s$ ($d_r$) is called the load that needs to be sent from sender server $s$ (received at receiver server $r$) for coflow $\calD$. Further, the effective size of the coflow is defined as
	\begin{equation}
	\label{CH3coflowload}
	D=\max \{\max_{s \in \calM_S} d_s , \max_{r \in \calM_R} d_{r}\}.
	\end{equation}
\end{definition}

Thus $D$ is the maximum load that needs to be sent or received by a server for the coflow. Note that, due to normalized capacity constraints on links, \textit{we need at least $D$ time slots to process all its flows}. 
\begin{definition}[Aggregate Size of a Set of Coflows] \label{CH3def4}
 Given a set of coflows, consider an aggregate coflow $\calD=\sum_{c}\calD^{c}$ for $c$'s in the set. Then, aggregate size of the set is defined as the effective size of $\calD$ based on Definition~\ref{CH3def1}. Similarly, aggregate size of job $j$ is defined as the aggregate size of its set of coflows and is denoted by $\Delta_j$ .
 \end{definition}
 
\begin{definition}[Size of a Directed Path and Critical Path in a Job] \label{CH3def2}
	Given a job $j \in \calN$ and its rooted tree $G_j$, size of a directed path $p$ in $G_j$ is defined as $T_{p,j}=\sum_{c \in p} D^{(cj)}$, where $D^{(cj)}$ is the effective size of coflow $c$ of job $j$, and $c \in p$ denotes that coflow $c$ appears in path $p$.
	
	Critical path of job $j$ is a directed path that has the maximum size among all the directed paths in $G_j$. We use $T_j=\max_{p} T_{p,j}$ to denote its size.
\end{definition}
\begin{definition}[A Path Job] \label{CH3def3}
	We say a job is a path job if its corresponding dependency graph is a path, i.e., there is a total ordering of its coflows according to which they should get scheduled.
\end{definition}
\begin{definition}[A Rooted-Tree Job] \label{rootedtreedef}
We say a job is a rooted-tree job if its corresponding dependency graph is a rooted tree, i.e., it is a tree and there is a unique node called \emph{the root} and either all the directed edges point away from this node (fan-out tree) or point toward this node (fan-in tree). For each rooted-tree job $G_j$, we use $R_j$ to denote its root.
\end{definition}
\begin{definition}[Height and Coflow Sets for a Job]\label{CH3def5}
Given a job $j \in \calN$ and its graph $G_j$, we define $H_j$ to be the height of $G_j$, i.e., the length of the longest path in $G_j$ (in terms of \textit{number} of coflows). Further, we define $S_0$ to denote the set of coflows with no in-edge. Similarly, define $S_i$, $i=1,\dots,H_j-1$ to denote the set of coflows whose longest path to some coflow of set $S_0$ has length $i$. Note that coflows in $G_j$ are partitioned by $S_i$s, i.e., $\cup_{i=0}^{H_j-1} S_i=G_j$ and $S_i \cap S_{i^\prime}=\varnothing$, for $i, i^\prime=0,\dots,H_j-1$, $i \neq i^\prime$. We refer to $S_i$s as coflow sets of job $j$.
\end{definition}

\subsection{Complexity of Minimizing Makespan}
\label{CH3complexity}
Scheduling a multi-stage job to minimize its completion time (makespan) is NP-hard. To show this, we consider a single multi-stage job whose DAG is a rooted tree. The proof is through a reduction from preemptive makespan minimization for Flow Shop Problem (FSP) which is known to be NP-complete~\cite{gonzalez1978flowshop,lawler1993sequencing,garey1976complexity}. This is in contrast to traditional coflow scheduling where a single coflow can be scheduled optimally as we see in Section~\ref{CH3mkspathjob}. This also shows that the known complexity results for preemptive FSP holds for single multi-stage job scheduling.
For FSP, there is no algorithm with an approximation ratio less than $5/4$, unless P = NP~\cite{williamson1997short}.

\begin{theorem}
	\label{CH3nphard}
	Given a single multi-stage job represented by a rooted tree, scheduling its coflows to minimize makespan over an $m \times m$ switch is NP-hard.
\end{theorem}
\begin{proof}
	We prove the theorem using a reduction from preemptive makespan minimization for Flow Shop Problem (FSP).
	In FSP, there is a set of $n$ jobs each of which consists of $m$ tasks that need to be processed \textit{in a given order} on $m$ machines. Task $i$ of job $j$ must be scheduled on machine $i$ for $p_{ij}$ amount of time (all the jobs require the same order on their tasks.). 
	Preemptive makespan minimization of FSP is known to be NP-complete~\cite{gonzalez1978flowshop,garey1976complexity}. 
	
	Consider an instance $I$ of FSP with $n$ jobs and $m$ machines. We convert the makespan minimization for $I$ to makespan minimization of an instance $I^\prime$ of a single multi-coflow job with a rooted tree topology. The instance $I^\prime$ consists of $m$ source and $m$ destination servers and $n \times m+1$ coflows where each has a single flow. Further, the corresponding dependency graph of $I^\prime$ is a tree with a root node and $n$ branches. The root node is a dummy coflow which has one flow of size one from source server $2$ (or any other source server) to destination server $1$. Each of the $n$ branches of the tree represents a job in $I$ and consists of $m$ coflows. The nodes in the $l$-th level of the tree, $l=1,\dots,m-1$ (the level of root node is zero) represent coflows that each has a single flow from source server $l$ to destination server $l+1$ with sizes $p_{lj}$, $j=1,\dots,n$. Similarly, the nodes at level $m$ are coflows with a single flow from source node $m$ to destination server $1$ with sizes $p_{mj}$, $j=1,\dots,n$. 
	
	If one can find the optimal makespan for the instance $I^\prime$ of a single multi-coflow job, the solution gives an optimal scheduling for the instance $I$ by ignoring the first time unit that is used to schedule the dummy coflow in $I^\prime$. Therefore, the theorem is proved. 
	\end{proof}

Using Theorem~\ref{CH3nphard} it is easy to see that minimizing makespan for multiple jobs and total weighted completion time of jobs are NP-hard. 
\subsection{Optimal Makespan for A Path Job}
\label{CH3mkspathjob}
In this section, we first show how one can schedule a single coflow optimally and in a polynomial time using the previous results. 
As a result of Birkhoff-von Neumann Theorem~\cite{birkhoff1946tres}, given a coflow $\calD=\big (d_{sr} \big )_{s,r=1}^m$ there exists a polynomial-time algorithm which finishes processing of all the flows in an interval whose length is equal to the coflow effective size $D$ (see Equation~\dref{CH3coflowload}). 

We present one example of such an algorithm in Algorithm~\ref{CH3Neumann}, which was proposed originally in~\cite{lawler1978preemptive}, and refer to it as {\BNA} that stands for Birkhoff-von Neumann Algorithm. {\BNA} returns a list of matchings $L$ and a list of times $\tau$. To schedule flows of $\calD$, we use each matching $L(k)$ for $\tau(k+1)-\tau(k)$ time units, for $k=1, \dots , |L|$.
 
%

\begin{algorithm}[!h]
	\captionof{algorithm}{{\BNA} for Single Coflow Scheduling}
	\label{CH3Neumann}
Given a coflow $\calD=\big (d_{sr} \big )_{s,r=1}^m$:
	\begin{enumerate}[leftmargin=*,label=\arabic*.]
		\item Let $L$ be the list of matchings and $\tau$ be the list of starting times for each matching. Initially, 
$L=\varnothing$, $\tau=[0]$.
			
		\item For any $s \in \calM_S$ and $r \in \calM_R$, compute $d_{s}$, $d_r$, and $D$ according to Definition~\ref{CH3def1}.
		\item Find the set of tight nodes as $\Omega=(\arg\max_{s \in \calM_s} d_{s}) \cup (\arg\max_{r \in \calM_R} d_{r}).$
		\item Find a matching $M$ among the source and destination nodes such that all the nodes in $\Omega$ are involved.
		\item Find 
		\ben
		t&=\min\Big\{ \min_{(s,r) \in M} d_{sr}, \min_{\substack{s: (s,r) \notin M}} (D-d_s),\\
		& \min_{\substack{r: (s,r) \notin M}} (D-d_r)\Big\}
		\een
		\item Add $M$ and $t+\tau [\text{end}]$ to the lists $L$ and $\tau$, respectively.
		\item Update the flow sizes as $d_{sr} \leftarrow d_{sr}-t,  \ \forall (s,r) \in M$. 
		\item While $\calD \neq \mathbf{0}$, repeat Steps $2-7$.
		\item Return $L$ and $\tau$.
	\end{enumerate}
\end{algorithm}
It is immediate that the optimal makespan for a path job can be found in polynomial time, by optimally scheduling its coflows successively using {\BNA}.

\begin{lemma}
	\label{CH3optmkspan}
	Optimal makespan for a path job $j$ is equal to $\sum_{c=1}^{\mu_j} D^{(cj)}$ where $D^{(cj)}$ is the effective size of coflow $c$ of job $j$ and the corresponding schedule can be constructed in polynomial time by  successively using {\BNA}.
\end{lemma}
We will use {\BNA} in our algorithms in the rest of the paper. 
\section{Makespan Minimization for Scheduling Multiple General DAG Jobs}
\label{CH3mkspan4multijobs0}
\subsection{{\DMA} (Delay-and-Merge Algorithm) }
\label{CH3mksgenjob}
For each job $j$, we consider a topological sorting of nodes in $G_j$, i.e., we sort its coflows (nodes) such that for every precedence constraint $c_1 \prec c_2$ (directed edge $c_1 \to c_2$), coflow $c_1$ appears before $c_2$ in the ordering. This ordering is not unique and can be found in polynomial time~\cite{knuth1997art}. For example, for the job in Figure~\ref{CH3jobcoflow}, the orderings $c_1, c_2, c_3, c_4,c_5,c_6, c_7$ and $c_2, c_3,c_1,c_5,c_4,c_6,c_7$ are both valid topological sorts. We then re-index coflows from $1$ to $\mu_j$ according to this ordering.

Further, we use $\Delta_j$ to denote the maximum load that a server should send or receive considering all of job $j$'s coflows. Formally, for job $j$, consider an aggregate coflow $\calD^j=\sum_{c=1}^{\mu_j} \calD^{(cj)}$. Then, $\Delta_j$ is the effective size of $\calD^j$ based on Definition~\ref{CH3def1}. We also use $\Delta$ to denote the maximum load a node has to send or receive \textit{considering all the jobs}. 

 Algorithm~\ref{CH3DMA} ({\DMA}) describes our algorithm for scheduling multiple general DAG jobs. 
\begin{algorithm}[!h]
	\captionof{algorithm}{{\DMA} for Scheduling a General DAG $G_j$}
	\label{CH3DMA}
	\begin{enumerate}[leftmargin=*,label=\arabic*.]
		\item For each job $j$, compute a topological sorting of nodes in $G_j$. Then, find a feasible schedule by optimally scheduling its coflows successively using {\BNA}, i.e., $L_{cj},\tau_{cj}=$ {\BNA}($\calD^{(cj)}$), for coflow $c=1, \dots, \mu_j $. We refer to these schedules as \textit{isolated} schedules of jobs. 
		
		\item Delay each isolated schedule by a random integer time chosen uniformly in $[0, \Delta/\beta]$, for a constant $\beta > 1/e$, independently of other isolated schedules, i.e.,  $\tau_{cj} \leftarrow \tau_{cj}+t_{j}$ where $t_{j}$ is the random delay of job $j$.
		\item Greedily merge the delayed isolated schedules. I.e., for any time slot $t$, add corresponding matchings of different jobs.
		\item Construct a feasible merged schedule. Let $\alpha_t \geq 1$ denote the maximum number of packets that a server needs to send or receive at time slot $t$ in the merged schedule in Step 3. For each time slot $t$, consider an interval of length $\alpha_t$, and use {\BNA} to feasibly schedule all its packets.
	\end{enumerate}
\end{algorithm}

Note that in {\DMA}, in each of the isolated schedules in Step 1 all the precedence constraints among coflows are respected. However, in Step 3, the link capacity constraints may be violated. In Step 4, in the final schedule, both link capacity constraints and precedence constraints among coflows are satisfied. The parameter $\beta >1/e$ in {\DMA} is a constant and has no effect on the theoretical result. However, it can be used to control the range of delays in practice.
 \begin{figure}[t]
	\centering
	\includegraphics[width=3.6in,height=1.6in]{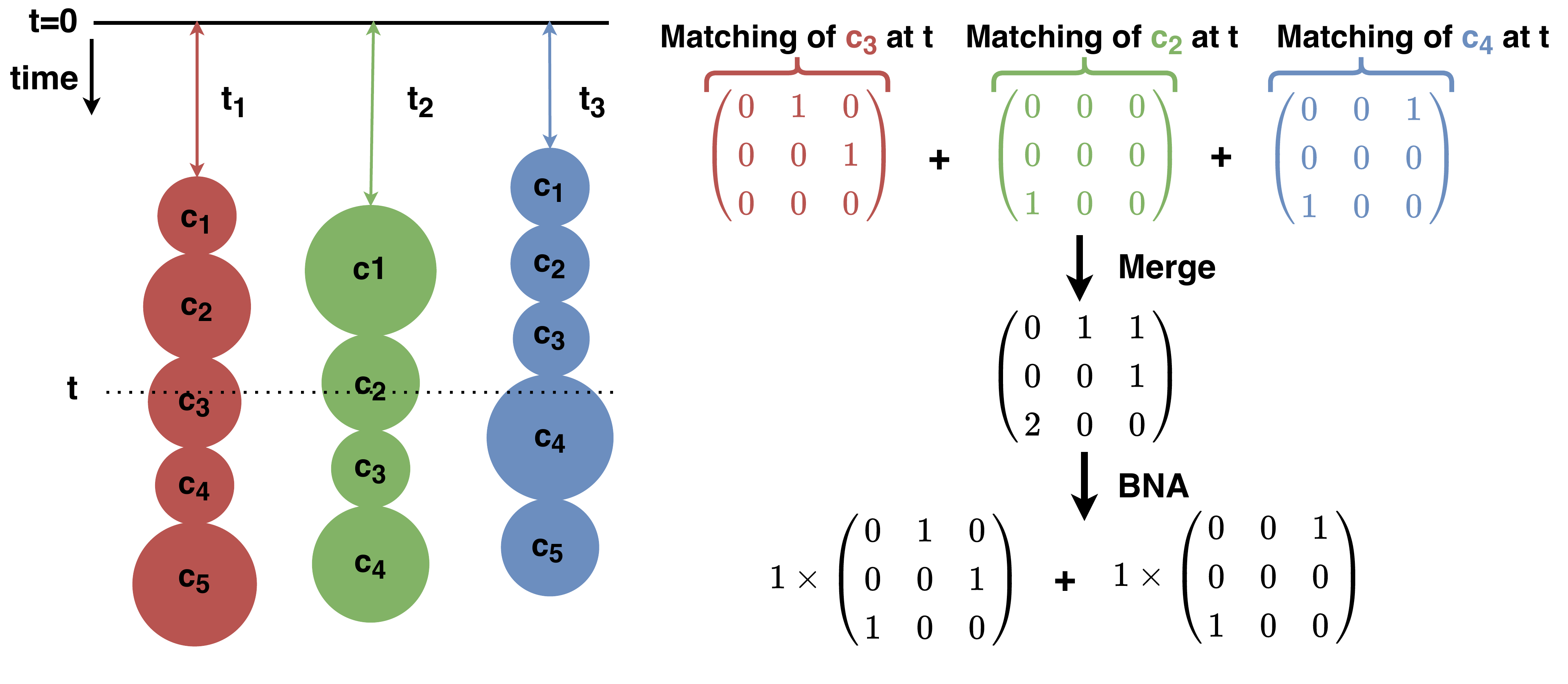}%
	\caption{Applying {\DMA} on $3$ multi-stage jobs. On the left side, a topological ordering and a random delay for each job are computed. On the right side, the merging procedure and {\BNA} output is shown for some time $t$.}
	\label{CH3DMAproc}%
\end{figure}

As an illustration, Figure~\ref{CH3DMAproc} shows the procedure of {\DMA} on 3 multi-stage jobs in a $3 \times 3$ switch network. On the left side, {\DMA} computes a topological ordering for the coflows of each job and chooses a random delay for each job. The diameter of each node is proportional to the effective size of its corresponding coflow. Consider time slot $t$,  {\DMA} merges the matchings of coflow $3$ of the red job, coflow $2$ of the green job, and coflow $4$ of the blue job, and inputs the result to {\BNA}. Then, {\BNA} computes two matchings, where each should be used for one time slot.

\subsection{Performance Guarantee of {\DMA}} The following theorem states the main result regarding the performance of {\DMA}. The proof can be found in Section~\ref{proofDMA}.
\begin{theorem}
	\label{CH3Thm:mksgenjob}
	Given a set $\calN$ of jobs with general DAGs, {\DMA} runs in polynomial time and provides a feasible solution whose makespan $\calT^{(\calN)}$ is at most $O(\mu g(m))$ of the optimal makespan with high probability, where $g(m)=\log(m)/\log(\log(m))$.
\end{theorem}

\subsection{De-Randomization}
Step 2 of {\DMA} involves random choices of delays. There exist well-established techniques that one can utilize to de-randomized this step and convert the algorithms to deterministic ones. For instance, one approach for selecting good delays is to cast the problem as a vector selection problem and then apply techniques developed in~\cite{raghavan1987randomized,raghavan1988probabilistic,shmoys1994improved}.

\section{Makespan Minimization For Scheduling Multiple Rooted Tree Jobs}
\label{CH3mkspan4multijobs}
Now we consider the case where each job is represented by a rooted tree (Definition~\ref{rootedtreedef}). We propose an algorithm with an improved performance guarantee compared to the case of general DAGs.
We would like to emphasize that the $O(m)$ approximation algorithm~\cite{tian2019scheduling,tian2018scheduling} will not be improved if the graph is a rooted tree rather than a general DAG.

\subsection{{\DMART} (Delay-and-Merge Algorithm For A Single Rooted Tree)}
\label{CH3mksgenjobRT}
In this section, we develop an approximation algorithm for minimizing makespan of a single rooted-tree job and show that its solution is at most $O(\sqrt{\mu}\log (m\mu)/\log(\log(m\mu)))$ of the optimal makespan. 
Recall Definitions~\ref{rootedtreedef} and~\ref{CH3def5}. In what follows, we assume that the rooted tree $G_j$ has an orientation towards the root $R_j$ (i.e. fan-in tree). For the case that edge orientations points away from the root (i.e. fan-out tree), the algorithm is similar.
Recall that $S_0$ is the set of coflows with no in-edge in rooted tree $G_j$. For each coflow $c \in S_0$, we can find a directed path   
starting from $c$ and ending at coflow (node) $R_j$. We call each of these paths a \textit{path sub-job} of job $j$. We use $\calP_j$ to denote the set of all path sub-jobs of job $j$. Recall that $T_{p,j}$ is the size of directed path $p\in \calP_j$ and $T_j$ is the size of the critical path (see Definition~\ref{CH3def2}). Figure~\ref{exRT} shows a rooted tree with $3$ path sub-jobs.

\begin{figure}[t]
	\centering
		\includegraphics[width=1 in,height=1.3in]{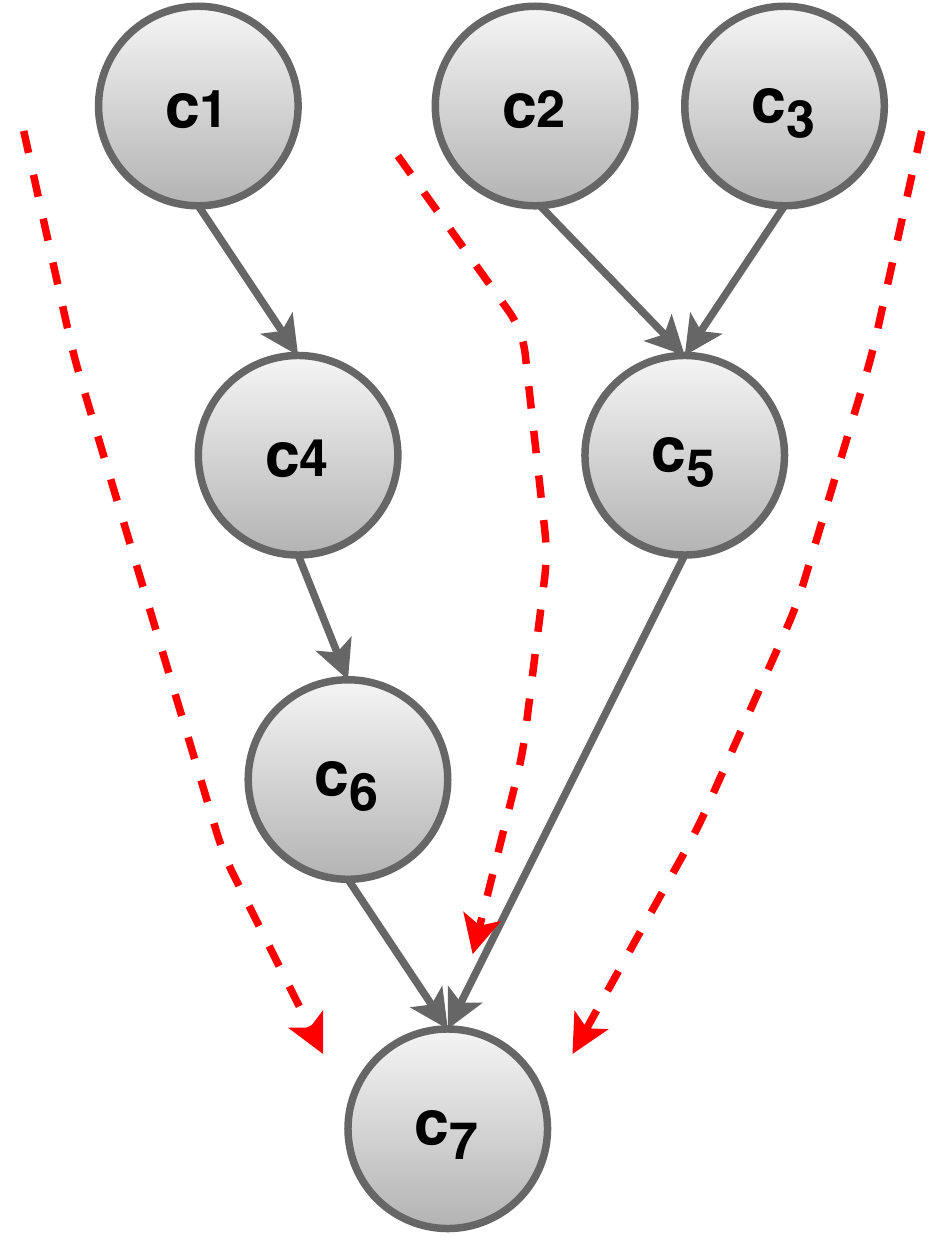}%
		\caption{A rooted tree with $3$ path sub-jobs.}
		\label{exRT}
\end{figure}
Algorithm~\ref{CH3alg:DMA} provides description of {\DMART}.
\begin{algorithm}[!h]
	\captionof{algorithm}{{\DMART} for Scheduling a Rooted Tree $G_j$}\label{CH3alg:DMA}
	\begin{enumerate}[leftmargin=*,label=\arabic*.]
		\item Find the set of path sub-jobs $\calP_j$ of job $j$. For each path sub-job $p \in \calP_j$, Choose a random integer time $d_p$ uniformly in $[0, \Delta_j/\beta]$, for a constant $\beta > 1/e$, independent of other isolated schedules. Next, for each coflow $c \in p$, $p \in \calP_j$,  calculate the starting time of coflow $c$ according to $p$, $t_{c,p}=d_p+\sum_{c^\prime \prec c, c^\prime \in p}D^{{c^\prime j}}$.
		\item Find the coflow sets $S_i$, $i=0, \dots, H_j-1$ of job $j$ according to Definition~\ref{CH3def5}. For $i=0, \dots, H_j-1$, and for each coflow $c$ in $S_i$, find starting time of coflow $c$ as $t_c=\min \{t_{c,p}| t_{c,p}\geq \max_{c^\prime \in \pi_c} (t_{c^\prime}+D^{(c^\prime j)})\}$. 
		\item For each coflow $c$ in $G_j$, find an optimal schedule for each coflow $c$ using {\BNA}, i.e., $L_c,\tau_c=${\BNA}($\calD^{(cj)}$). We refer to these schedules as \textit{isolated} schedules. Then, delay the scheduling times by $t_c$, $\tau_c \leftarrow \tau_c+t_c$.
		\item Follow Step 3 of {\DMA}.
		\item Follow Step 4 of {\DMA}.
			
	\end{enumerate}
\end{algorithm}

Note that the algorithm calculates the starting time of each coflow, $t_c$, such that all the precedence constraints of the coflow are satisfied. In other words, $t_c$ is equal to the smallest time $t_{c,p}$ (starting time of $c$ based on path $p$) that all its preceding coflows in $G_j$ are completed. We say that $c$ is scheduled according to $p$ if $t_c=t_{c,p}$.
Therefore, the merged schedule satisfies all the precedence constraints among coflows, although the link capacity constraints may be violated. {\DMART} constructs a feasible merged schedule using {\BNA}. Note that in Step 5,  $\calD$ is multiplied by $l_I$ since each matching $L_c(i)$ runs for $l_I$ time units in its corresponding isolated schedule. In the final schedule, both link capacity constraints and precedence constraints among coflows are satisfied.

\subsection{Multiple Rooted Tree Jobs}
Now consider the case where we have multiple jobs where each job is a rooted tree. We seek to find a feasible schedule that minimizes the time to process all the jobs (makespan). Recall that $\mu$ is the maximum number of coflows in any job. We use $\Delta$ to denote the aggregate size of coflows of \textit{all the jobs} (Definition~\ref{CH3def4}). 

The scheduling algorithm is based on {\DMART} described in Section~\ref{CH3mksgenjobRT}. Specifically, we apply {\DMART} to find a feasible schedule for each job in the set. Then we apply Steps 2, 3 and 4 of {\DMA}, namely, we choose a random delay in $[0, \Delta/\beta]$ for a constant $\beta > 1/e$ for each individual schedule and delay it. Next, we merge the delayed schedules. Finally we use {\BNA} algorithm to resolve any collisions in the merged schedule. We refer to this algorithm as {\OCDMA}.

\subsection{Performance Guarantee of {\DMART} and {\OCDMA}} 
\label{OCDMAproofs}
\begin{theorem}
	\label{CH3Thm:mksgenjobRT}
	Given a single job $j$ with rooted tree $G_j$, {\DMART} runs in polynomial time and provides a feasible schedule whose makespan $C_j$ is at most $O(\sqrt{\mu_j} h(m,\mu_j))$ of the optimal makespan with high probability, where $h(m,\mu)= \log (m\mu)/\log(\log(m\mu))$.
\end{theorem}
\begin{theorem}
	\label{CH3Thm:mksmultigenjobs}
	Given a set $\calN$ of jobs, each represented as a rooted tree, {\OCDMA} runs in polynomial time, and achieves a solution whose makespan $\calT^{(\calN)}$ is at most $O(\sqrt{\mu} g(m) h(m,\mu))$ of the optimal makespan with high probability.
\end{theorem}

The proofs of Theorems~\ref{CH3Thm:mksgenjobRT} and~\ref{CH3Thm:mksmultigenjobs} are presented in Section~\ref{proofDMART}.

\section{Total Weighted Completion Time Minimization}
\label{CH3weightedminsum}
We are now ready to present our combinatorial approximation algorithm for minimizing the total weighted completion time of multi-stage jobs with release times. \textit{In this section, we assume that the jobs have general DAGs, however, the results can be customized for the case that all the jobs are represented by rooted trees.} We use $\rho_j$ to denote the release time of job $j$, which implies that  job $j$ is available for scheduling only after time $\rho_j$. 
\subsection{Job Ordering}
\label{CH3coflowscheduling}
To formulate a relaxed linear program for our problem,
we note that if we ignore the precedence constraints among coflows of a job and aggregate all its coflows, we obtain a single-stage job (a coflow), and our problem is reduced to traditional coflow scheduling problem~\cite{qiu2015minimizing,ahmadi2017scheduling,shafiee2018improved,shafiee2017scheduling}.

Here, we use an LP formulation for such constructed single-stage jobs, but with an extra constraint for each job which roughly captures the barrier constraints among its coflows.
Formally, for each job $j$, consider the aggregate coflow $\calD^j=\sum_{c=1}^{\mu_j} \calD^{(cj)}$.
Let $\overline{\calM} :=\calM_S \cup \calM_R$. We use $d_i^j$, $i \in \overline{\calM}$ to denote the 
load of coflow $\calD^j$ on server $i$ (see Definition~\ref{CH3def1}). Recall Definition~\ref{CH3def2} and note that $T_j$ is the lower bound on the required time to schedule multi-stage job $j$ (in the original problem). Let $\calJ$ be any subset of jobs in $\calN$. 
We formulate the following LP (Linear Program):
\begin{subequations}
	\label{CH3RLP}
	\begin{align}
	\label{CH3LPobj}
	\min  &\sum_{j \in \calN} w_j C_j \ \ \ \ \ \mathbf{(LP)}\\
	\label{CH3scap}
	& \sum_{j \in \calJ} d_i^j C_j \geq \frac{1}{2} \big(\sum_{j \in \calJ} (d_i^j)^2 + (\sum_{j \in \calJ} d_i^j)^2 \big) ,\ i \in \overline{\calM}, \calJ \subseteq \calN\\
	\label{CH3rdandcrp}
	& C_j \geq T_j+\rho_j, \ \  j \in \calN. 
	\end{align}
\end{subequations}
Constraints~\dref{CH3scap} capture the links' capacity constraints and are used to lower-bound the completion time variables. To see this, consider a (source or destination) server $i$ and a subset of jobs $\calJ$. For each $j$ in $\calJ$, the completion time $C_j$ of its aggregate coflow $\calD^j$, has to be at least the summation of loads of coflows $\calD^{j^\prime}$ on  server $i$ that finish before $j$ plus its own load on server $i$. Also note that for every two coflows in the set $\calJ$, one finishes before the other one. Therefore, $\sum_{j \in \calJ} d_i^j C_j \geq \sum_{j \in \calJ} d_i^j (d_i^j+\sum_{j^\prime \in \calJ, j^\prime \prec j} d_i^{j^\prime})$, where $j^\prime \prec j$ means $C_{j^\prime} \leq C_j$. From this, Constraint~\dref{CH3scap} is derived easily. 

Note that this LP has \textit{exponentially many constraints}, since we need to consider all the subsets of $\calN$. However, we do not need to explicitly solve this LP and we only need to find an efficient ordering of jobs. 
To do so, we utilize the combinatorial primal-dual algorithm that first proposed in~\cite{mastrolilli2010minimizing} and later generalized in~\cite{ahmadi2017scheduling} to capture constraints of the form~\dref{CH3rdandcrp} for parallel scheduling problems. The algorithm builds up a permutation of the jobs
in the reverse order iteratively by changing the corresponding dual variables to satisfy some dual constraint. We have provided the detailed explanation of the combinatorial algorithm in Appendix~\ref{app1} for completeness. We show how we use this ordering to find the actual schedule of jobs' coflows in the next section.

\begin{remark}
Algorithm~\ref{CH3permalg} in Appendix~\ref{app1} runs in $O(n(\log(n)+m))$ time where $n$ is the number of jobs and $m$ is the number of servers. However, the time complexity of the best known algorithm for solving the LP used in~\cite{tian2019scheduling,tian2018scheduling} is $O((n^2+m)^{\omega} \log((n^2+m)/\epsilon))$, where $\omega$ is the exponent of matrix multiplication and $\epsilon$ is the relative accuracy~\cite{cohen2019solving,van2020deterministic}. For current value of $\omega = 2.38$~\cite{williams2012multiplying,le2014powers}, the time complexity of Algorithm~\ref{CH3permalg} is dramatically lower than the time complexity for solving the LP used in~\cite{tian2019scheduling,tian2018scheduling}.
\end{remark}

\subsection{Grouping Jobs}
\label{CH3partitionrule}
Let $D_j$ denote the maximum load that a server has to send or receive considering all coflows of the jobs up to and including job $j$ according to the computed ordering. In other words, $D_j$ is the effective size of an aggregate coflow constructed from coflows of the first $j$ jobs. Recall that $T_j$ is size of the critical path in job $j$ (Definition~\ref{CH3def2}).
Define $\gamma=\min_{s,r,c,j} d_{sr}^{cj}$ which is a lower bound on the time required to process any job. Also let $T=\max_j \rho_j+\sum_{j \in N} \sum_{c \in j} \sum_{s \in M_S} \sum_{r \in M_R} d_{sr}^{cj}$. The algorithm groups jobs into $B$ groups as follows.

Choose $B$ to be the smallest integer such that $\gamma 2^{B} \geq T$, and consequently define
\be \label{CH3eq:partition rule}
a_b=\gamma 2^{b}, \mbox{ for } b=-1,0,1,...,B.
\ee
Then the $b$-th interval is defined as the interval $(a_{b-1},a_b]$ and the group $\calJ_b$ is defined as the subset of jobs whose $T_j+\rho_j+D_j$ fall within the $b$-th group, i.e.,
\be \label{CH3eq:subset}
\calJ_b=\{j \in N:  T_j+\rho_j+D_j  \in (a_{b-1},a_b]\};\ 0 \leq b \leq B.
\ee
This partition rule ensures that every job falls in some group.
\subsection{Scheduling Each Group $\calJ_b$}
To schedule jobs of each group $\calJ_b$, $b\in\{1,\cdots,B\}$, (defined by \dref{CH3eq:subset}), we use the {\DMA} algorithm. We refer to this algorithm as {\OJOD} algorithm which stands for \textit{Grouping jobs, followed by Delay-and-Merge} algorithms. We summarize {\OJOD} in Algorithm~\ref{CH3minsumalg}.
\begin{algorithm}[!h]
	\captionof{algorithm}{{\OJOD} for Scheduling Multi-Stage Jobs}
	\label{CH3minsumalg}

\begin{enumerate}[leftmargin=*,label=\arabic*.]
	\item Find an efficient permutation of jobs using Algorithm~\ref{CH3permalg} and re-index them.
	\item Let $D_j$ be effective size of the \emph{aggregate} coflow constructed from coflows of the jobs up to and including job $j$. Also, let $T_j$ be size of the critical path in job $j$.
	\item Partition jobs into disjoint subsets $\calJ_b$, $b=0,...,B$ as in~\dref{CH3eq:subset}.
	\item For each group $b=1, \dots,B$, wait until all jobs in $\calJ_b$ arrive, then apply the makespan minimization algorithm {\DMA} to schedule them.
\end{enumerate}
\end{algorithm}
\subsection{Performance Guarantee of {\OJOD}}
Recall that $g(m)=\log(m)/\log(\log(m))$, and $h(m, \mu)=\log(m\mu)/(\log(\log(m\mu))$. The following theorem states the main result regarding the performance of {\OJOD}.

\begin{theorem}
	\label{CH3minsumalgper}
	{\OJOD} is a polynomial-time $O(\mu g(m))$-approximation algorithm for the problem of total weighted completion time minimization of multi-coflow jobs with release dates.
\end{theorem}

For the case that we are given a set $\calN$ of jobs, each represented as a rooted tree, we modify {\OJOD} by using {\OCDMA} as the subroutine in the last step of {\OJOD}. We denote the modified version as {\GOCDM}. We then have the following Corollary.
\begin{corollary}
 {\GOCDM} is a polynomial-time algorithm with approximation ratio $O(\sqrt{\mu} g(m) h(m,\mu))$ for minimizing the total weighted completion time of rooted-tree jobs with release times.
\end{corollary}

The proofs can be found in Section~\ref{proofGDM}.
\section{Empirical Evaluation}
\label{simulation}
To demonstrate the gains in practice, we conducted extensive evaluations using a real workload. This workload has been widely used in coflow related research~\cite{Chowdhury2014, qiu2015minimizing, tian2018scheduling, shafiee2018improved}. We compared the performance of our algorithm {\GOCDM} with the $O(m)$-algorithm in~\cite{tian2019scheduling,tian2018scheduling} which is the previous state-of-the-art algorithm and compare its performance with that of our algorithm. In~\cite{tian2019scheduling,tian2018scheduling}, the authors have shown that their algorithm outperforms single-stage coflow scheduling algorithms by around $83\%$, and Aalo~\cite{chowdhury2015efficient} by up to $33\%$ for the case of equal weights for job (as Aalo cannot handle the weighted scenario). Hence, we only report comparison with this algorithm. The results indicate that our algorithm outperforms the $O(m)$-algorithm~\cite{tian2019scheduling,tian2018scheduling} by up to $53 \%$ \textit{in the same settings}. We also investigate the performance of the algorithms for different values of delaying parameter $\beta$, and problem size $\mu$ and $m$.

\begin{figure}[t]
	\centering
	\includegraphics[width=2.3in,height=1.6in]{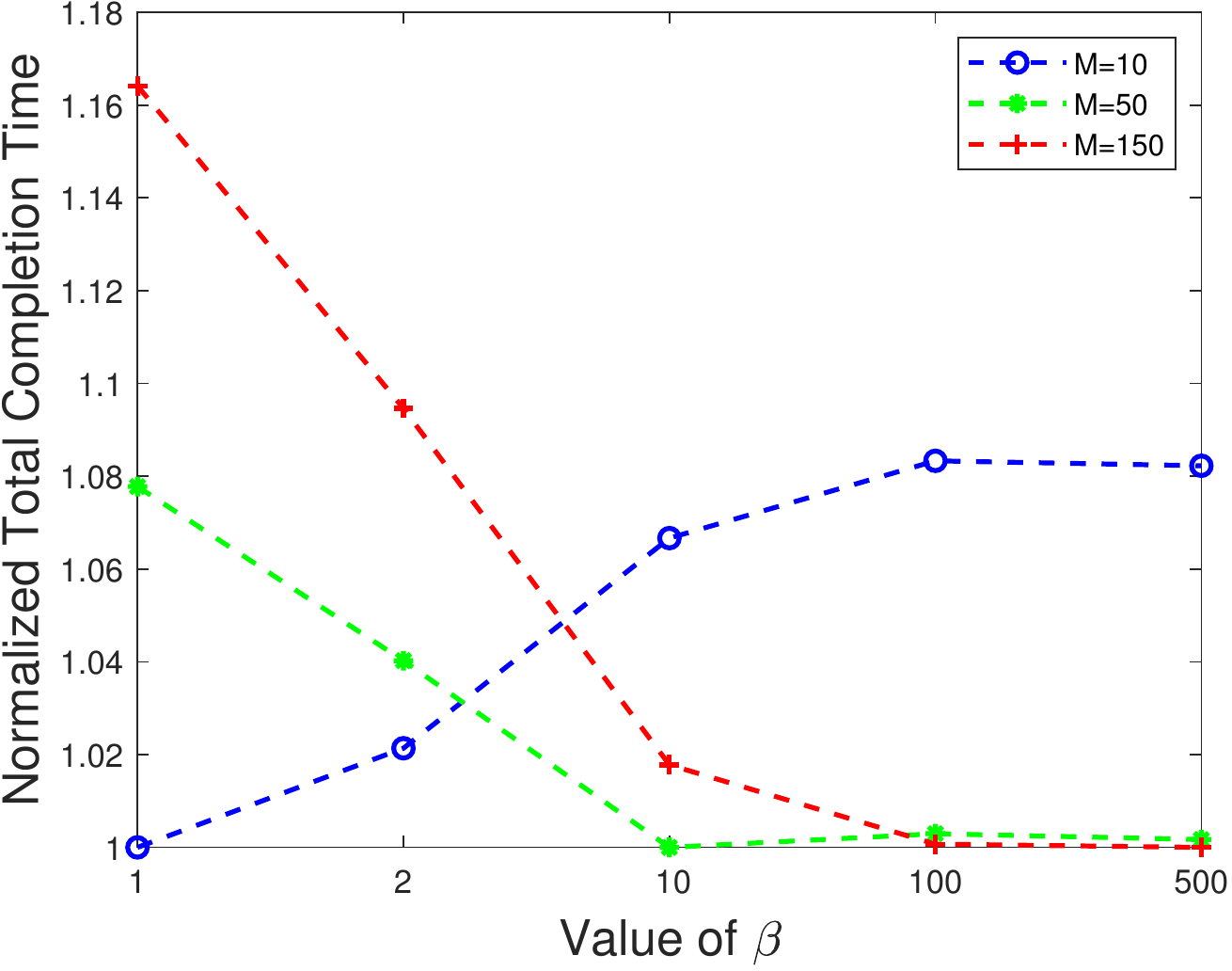}%
	\caption{Performance of {\GOCDM} for different  number of servers and different values of  $\beta$, and $\bar{\mu}=5$.}
	\label{effectofbeta}%
\end{figure}

\begin{figure*}[h]
	\centering
	\begin{subfigure}{0.3\textwidth}
		\includegraphics[width=\textwidth]{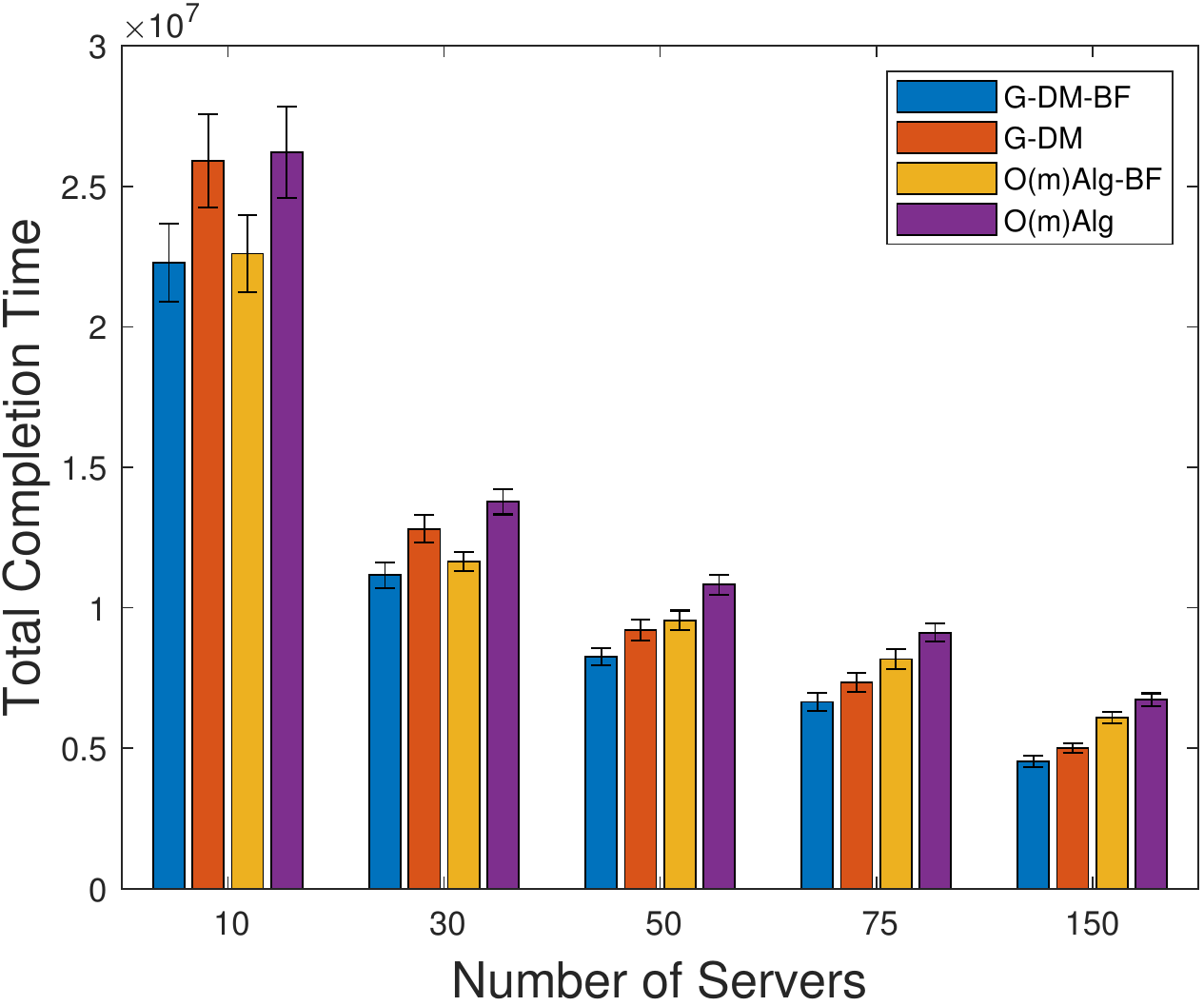}%
		\vspace{4pt}
		\caption{Performance of {\OJOD} and {\textsf{O(m)Alg}} with and without backfilling for different numbers of servers, and $\bar{\mu}=5$.}
		\label{CH3GDMserver}
	\end{subfigure}\hfill
	\begin{subfigure}{0.3\textwidth}
		\vspace{8pt}
		\includegraphics[width=\textwidth,height=1.8 in]{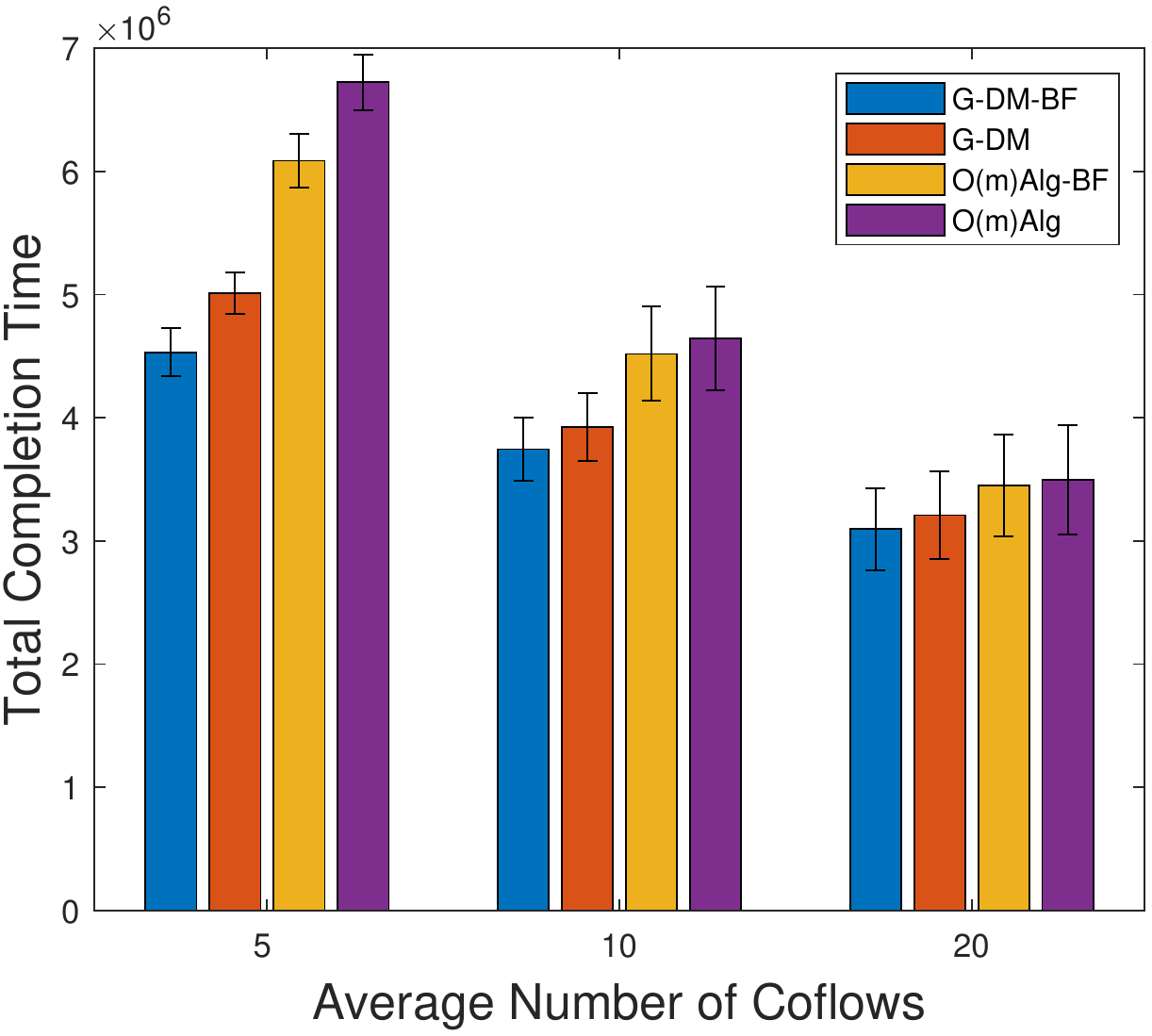}%
		\vspace{3.5pt}
		\caption{Performance of {\OJOD} and {\textsf{O(m)Alg}} with and without backfilling for different average numbers of coflows per job, and $m=150$.}
		\label{CH3GDMcoflow}%
	\end{subfigure}\hfill
	\begin{subfigure}{0.31\textwidth}
		\vspace{-2pt}
		\includegraphics[width=\textwidth,height=1.78 in]{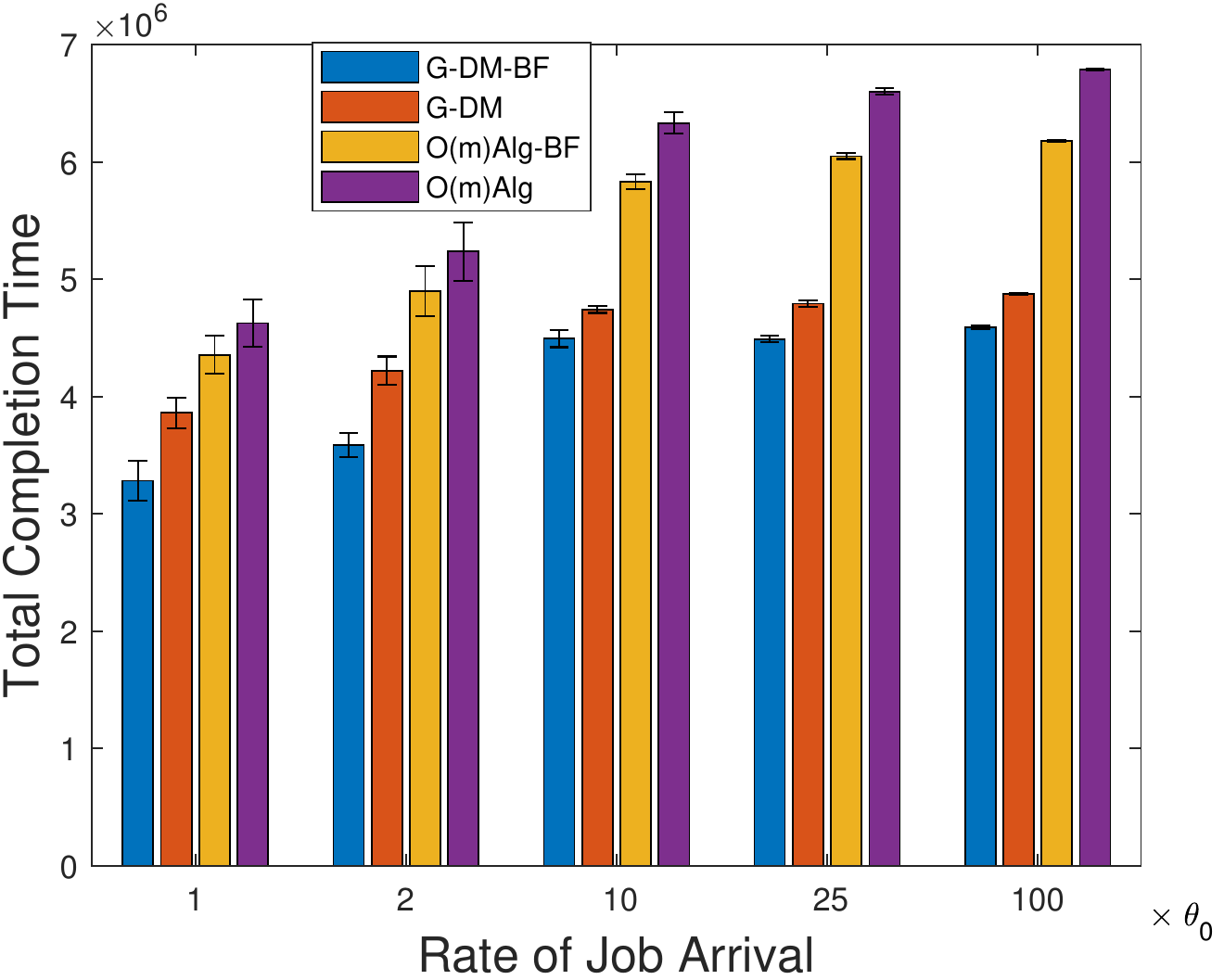}%
		\vspace{5pt}
		\caption{Performance of {\OJOD} and {\textsf{O(m)Alg}} with and without backfilling for different arrival rates, and $\bar{\mu}=5$, $m=150$.}
		\label{CH3GDMolt}%
	\end{subfigure}
	\caption{Performance of {\OJOD} and {\textsf{O(m)Alg}} for scheduling general DAGs with and without backfilling.}
\end{figure*}

\begin{figure*}[h]
	\centering
	\begin{subfigure}{0.3\textwidth}
		\includegraphics[width=\textwidth]{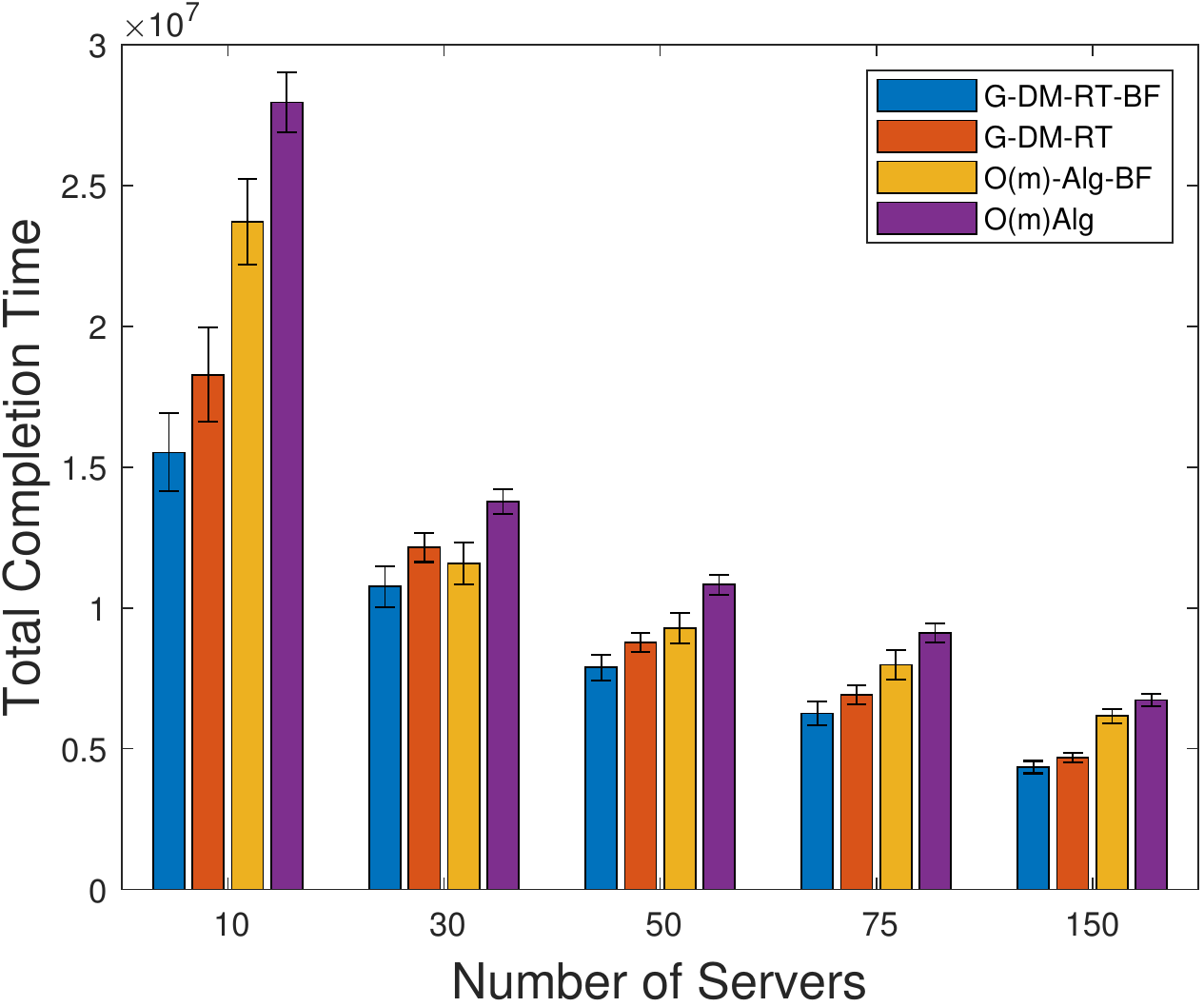}%
		\vspace{4pt}
		\caption{Performance of {\GOCDM} and {\textsf{O(m)Alg}} with and without backfilling for different numbers of servers, and $\bar{\mu}=5$.}
		\label{CH3equalweightserver}
	\end{subfigure}\hfill
	\begin{subfigure}{0.3\textwidth}
		\vspace{8pt}
		\includegraphics[width=\textwidth,height=1.8 in]{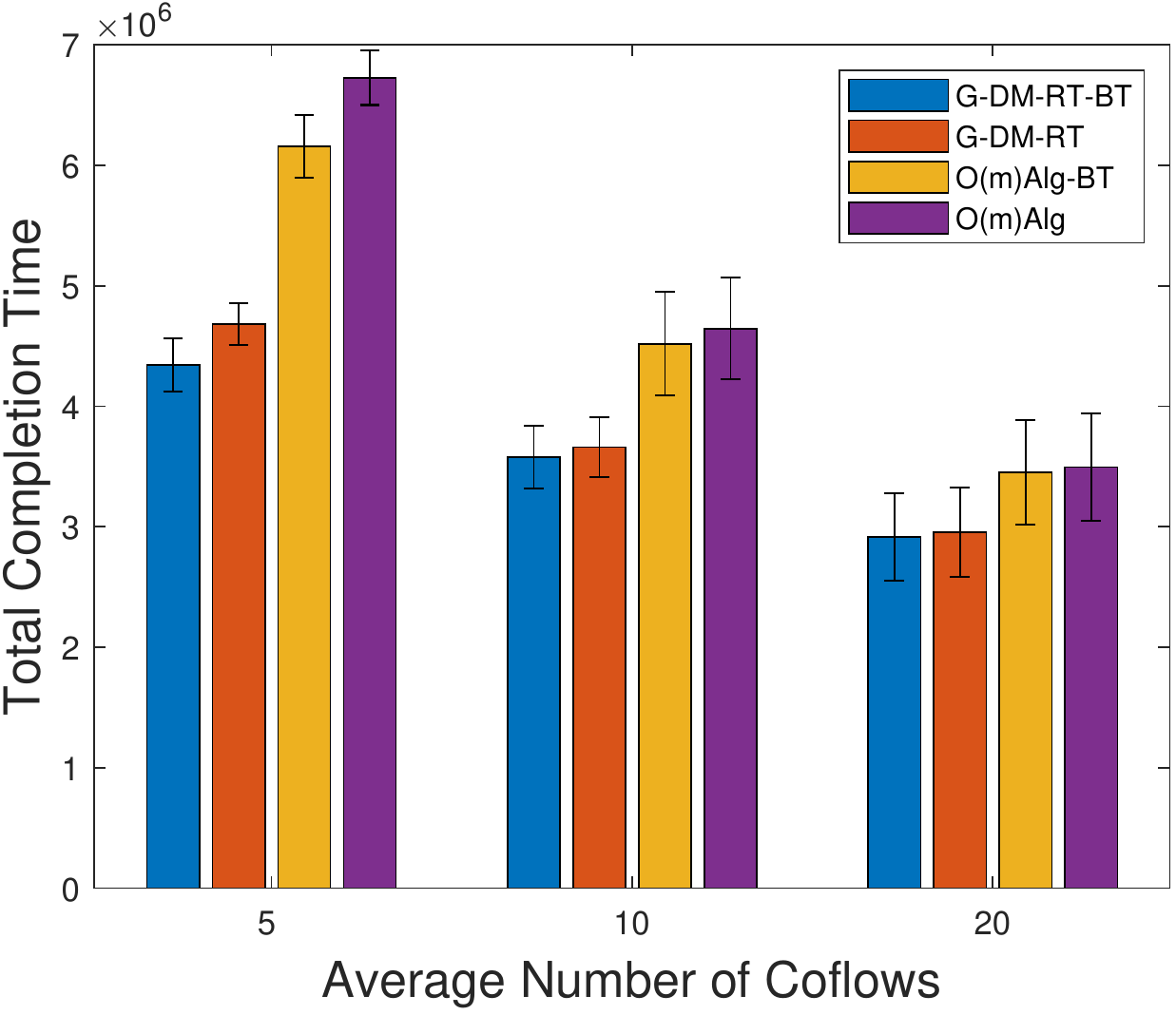}%
		\vspace{3.5pt}
		\caption{Performance of {\GOCDM} and {\textsf{O(m)Alg}} with and without backfilling for different average numbers of coflows per job, and $m=150$.}
		\label{CH3equalweightcoflow}%
	\end{subfigure}\hfill
	\begin{subfigure}{0.31\textwidth}
		\vspace{-2pt}
		\includegraphics[width=\textwidth,height=1.78 in]{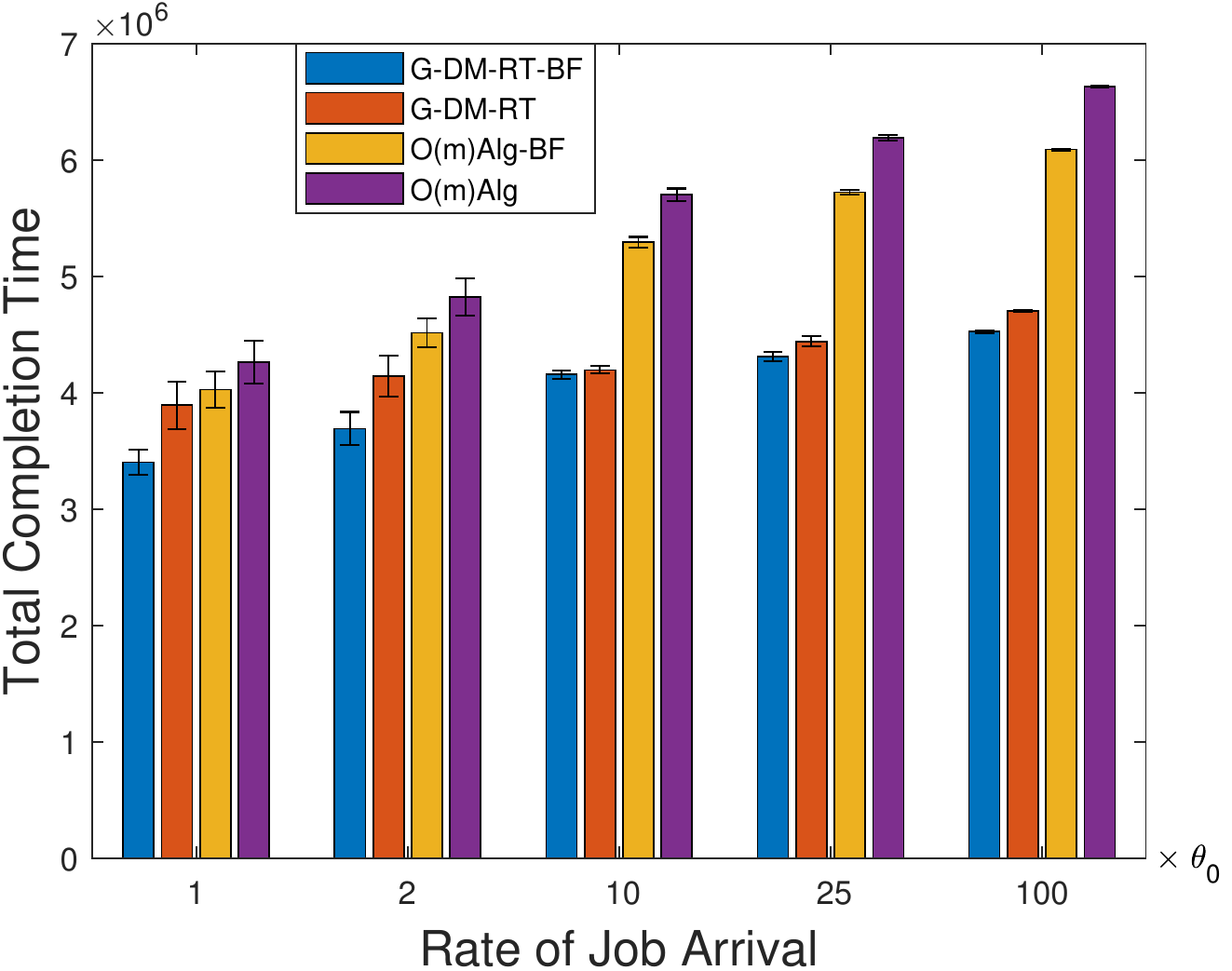}%
		\vspace{5pt}
		\caption{Performance of {\GOCDM} and {\textsf{O(m)Alg}} with and without backfilling for different arrival rates, and $\bar{\mu}=5$, $m=150$.}
		\label{CH3olequalweight}%
	\end{subfigure}
	\caption{Performance of {\GOCDM} and {\textsf{O(m)Alg}} for scheduling rooted tree jobs with and without backfilling.}
\end{figure*}

{\bf{Workload:}} The workload is based on a Hive/MapReduce trace at a Facebook cluster with $150$ racks, and only contains coflows information. The data set contains $267$ coflows with $\mu_j$ ranging from $10$ to $21170$. Further, size of the smallest flow is equal to $\gamma=1$, size of the largest flow is equal to $2472$, and effective size of coflows, $\Delta_j$, is between $5$ and $232145$. Finally, the maximum load a server should send or receive considering all the coflows, i.e., the effective size of the aggregate coflow, is equal to $\Delta=440419$.

To assess performance of algorithms under different traffic intensity, we generate workloads with different number of machines (servers) by mapping flows of the original $150$ racks to $m$ machines with various values of $m$. To generate multi-stage jobs, we randomly partition the coflows into multi-stage jobs that each has $\overline{\mu}$ coflows on average. To generate the corresponding rooted tree, we first generate a random graph in which probability of picking each of the edges is $0.5$, and then converting it to a tree by removing its cycles. We ran the algorithms for two cases of equal weights for all jobs and randomly selected weights from interval $[0,1]$. We also consider the online scenario where multi-stage jobs arrive over time and their release (arrival) times follow a Poisson process with a parameter $\theta$.

{\bf{Algorithms:}} We simulate our multi-stage job algorithms (referred to as {\OJOD} and {\GOCDM}) and the algorithm in~\cite{tian2019scheduling,tian2018scheduling} (referred to as {\textsf{O(m)Alg}}). 
For each algorithm, we present two versions, one with no backfilling and one with backfilling. Backfilling is a common technique in scheduling to increase utilization of system resources by allocating the underutilized link capacities (or servers, depending on the problem) to other jobs. We apply the same backfilling strategy to both algorithms for a fair comparison.
We use {\OJODBF}, {\GOCDMBF}, and {\textsf{O(m)Alg-BF}} to refer to the versions of algorithms with backfilling. 

{\bf{Metrics:}} We compare the  total weighted completion times of jobs under the two algorithms for various workloads and scenarios. We present results for offline and online scenarios with equal and random job weights. We also investigate the performance of the algorithms for different values of $m$, $\bar{\mu}$, $\theta$.

\subsection{Impact of Random Delays and $\beta$} 
The current implementations of {\OJOD} and {\GOCDM} have a random component as it uses {\DMA} and {\DMART} as a subroutine. To show that in practice running the algorithm \emph{once} is sufficient to achieve a satisfactory solution, we need to show that its relative standard deviation (RSD) is small. RSD is defined as standard deviation divided by the mean (average). Hence, to analyze the effect of random delays in the performance of our algorithm, we ran it on some instances, each for $10$ times.
Based on our experiments, RSDs of {\OJOD} and {\GOCDM} are always less than $0.5\%$ and RSDs of {\OJODBF} and {\GOCDMBF} are always less than $0.9\%$, which both are very small. In the rest of simulations, we run our algorithms only once on each instance.

Furthermore, we studied the effect of parameter $\beta$ (see Sections~\ref{CH3mkspan4multijobs0} and~\ref{CH3mkspan4multijobs}) on the performance of our algorithms. For each algorithm, we ran the algorithm using a wide range of $\beta$ values. 
Based on our experiments, for smaller $m$ (higher traffic intensity) it is better to choose a small value of $\beta$ ($1$ or $2$) to reduce the collision probability (\ref{CH3kepsilon}), while choosing larger $\beta$ ($100$ or $500$) for larger $m$ helps the algorithm to use the unused capacity to schedule flows of other coflows in the system. Moreover, the amount of improvement by optimizing over $\beta$ was less than  $16\%$ in all the experiments. Figure~\ref{effectofbeta} shows the results for different values of $\beta$ and $m$ when $\bar{\mu}$ is set to $5$ for {\GOCDM}.

\subsection{Evaluation Results for General GADs}
\subsubsection{Offline Setting} In the offline scenario, all the jobs are available at time $0$. For each set of parameters $(m,\bar{\mu})$, we generate $10$ different instances randomly and report the average and standard deviation of each algorithm's performance. 

Figure~\ref{CH3GDMserver} and~\ref{CH3GDMcoflow} depict some of the results for the case that jobs have general DAGs and equal weights. Figure~\ref{CH3GDMserver} shows the performance of {\OJOD} and {\textsf{O(m)Alg}} for the case that  average number of coflows per job, $\bar{\mu}$, is $5$ and different number of servers. {\OJOD} performs as well as {\textsf{O(m)Alg}} for $m=10$. It outperforms {\textsf{O(m)Alg}} from $9\%$ for $m=30$ to about $36\%$ for $m=150$. Moreover, Figure~\ref{CH3GDMcoflow} shows that \textit{our algorithm outperforms {\textsf{O(m)Alg}} for all values of average coflows per job, by $36\%$ to $11\%$}. The results for the case of random job weights are very similar and omitted.

\subsubsection{Online Setting} For the online scenario, jobs arrives to the system according to a Poisson process with rate $\theta$. Every time that a job arrives both {\OJOD} ({\GOCDM}) and {\textsf{O(m)Alg}} suspend the previously active jobs, update the list of jobs and their remaining demands, and reschedule them. Moreover, completion time of a job in the online scenario is measured from the time that the job arrives to the system. The job arrival rate is determined as follows: $\theta=a \times \theta_0$ for $a=\{1,2,10,25,100\}$, and $\theta_0=\frac{\sum_j \mu_j}{\sum_j \sum_c D^{cj}}$, in which $\sum_j \mu_j$ is the total number of coflows among all jobs. The denominator, $\sum_j \sum_c D^{cj}$, is summation of coflows' effective sizes and an upper bound on the jobs' makespan.

Figure~\ref{CH3GDMolt} shows the results under {\OJOD}  and {\textsf{O(m)Alg}} for the case that $m=150$ (original data set), $\bar{\mu}=5$, and all the jobs have equal weights. {\OJOD} always outperforms {\textsf{O(m)Alg}}, from $20\%$ to $36\%$. 
Furthermore, \textit{{\GOCDMBF} always outperforms {\textsf{O(m)Alg-BF}}, by $30\%$ to $37\%$.}

\subsection{Evaluation Results for Rooted Trees}
Now we provide the simulation results for the case that all the jobs are rooted trees.
\subsubsection{Offline Setting} Figure~\ref{CH3equalweightserver} shows the performance of two algorithms for different number of servers, $\bar{\mu}=5$, and equal weights for jobs. As we can see, {\GOCDM} always outperforms {\textsf{O(m)Alg}}, for about $53\%$ for $m=10$ to about $46\%$ for $m=150$.  For all values of average coflows per jobs, \textit{our algorithm outperforms {\textsf{O(m)Alg}} , by $46\%$ to $18\%$} as depicted in Figure~\ref{CH3equalweightcoflow}.

\subsubsection{Online Setting} Figure~\ref{CH3olequalweight} shows the results with and without backfilling for the case that $m=150$ (original data set), $\bar{\mu}=5$, and all the jobs have equal weights. {\GOCDM} always outperforms {\textsf{O(m)Alg}}, from $10\%$ to $46\%$.  
Furthermore, \textit{{\GOCDMBF} always outperforms {\textsf{O(m)Alg-BF}}, by $22\%$ to $36\%$.}

We would like to point out that, as we expect, the gain under {\GOCDM} is greater than {\OJOD}, as the former algorithm utilizes the network resources more efficiently by interleaving schedules of different coflows of the \emph{same job} as well as interleaving schedules of coflows of \emph{different jobs}. Furthermore, backfilling strategy generally yields a larger improvement when combined by {\OJOD} and {\textsf{O(m)Alg}} compared to {\GOCDM}, as they leave more resources unused.
\section{Discussion on Approximation Results}
\label{CH3discussion} 
An interesting research direction is to improve the approximation ratios for the algorithms. As we showed in the previous sections, once we have an algorithm for scheduling a single job whose solution is a factor $\eta$ of the simple lower bounds $\Delta_j$ and $T_j$ (Definitions~\ref{CH3def4} and~\ref{CH3def2}), we can directly utilize the rest of our approach and get approximation algorithms with approximation ratio $O(\eta \log(m)/\log\log(m))$ for the problems of makespan minimization and total weighted completion time  minimization for multiple jobs.

To improve the result for the case of general DAGs, one approach is to first consider scheduling a single job (with a general DAG), and try to generalize {\DMART} to a general DAG by careful construction of paths in the algorithm, so that we do not need to consider all the paths in the DAG which could be exponentially many. However, even if one could show that $O(\mu_j)$ paths is sufficient to construct a feasible schedule, it is challenging to analyze the performance through computing the probability of collisions or the average number of collisions in the merged schedule as we did in proof of Lemma~\ref{CH3numcolli1}. 
This is due to the underlying dependency among the unrelated coflows in $G_j$ (these are coflows among which there is no directed path in $G_j$, thus they can collide) which appears in the probability that a given coflow is assigned to start scheduling at a given time given that a specific $O(\mu_j)$ set of paths is generated by the algorithm.

Besides these challenges for scheduling a job with a general DAG, we can illustrate the existence of instances for which the optimal makespan is $\Omega(\sqrt{\mu_j})$ factor larger than the two simple lower bounds, $\Delta_j$ and $T_j$. We state this in the following lemma.
\begin{lemma}
\label{CH3tightness}
There exist arbitrary sized instances of DAG job scheduling such that its optimal makespan is $\Omega(\sqrt{\mu_j}(\Delta_j+T_j))$.
\end{lemma}
\begin{figure}
	\centering
	\begin{subfigure}{0.5\columnwidth}
	\centering
	\includegraphics[width=1.2 in,height=1.6in]{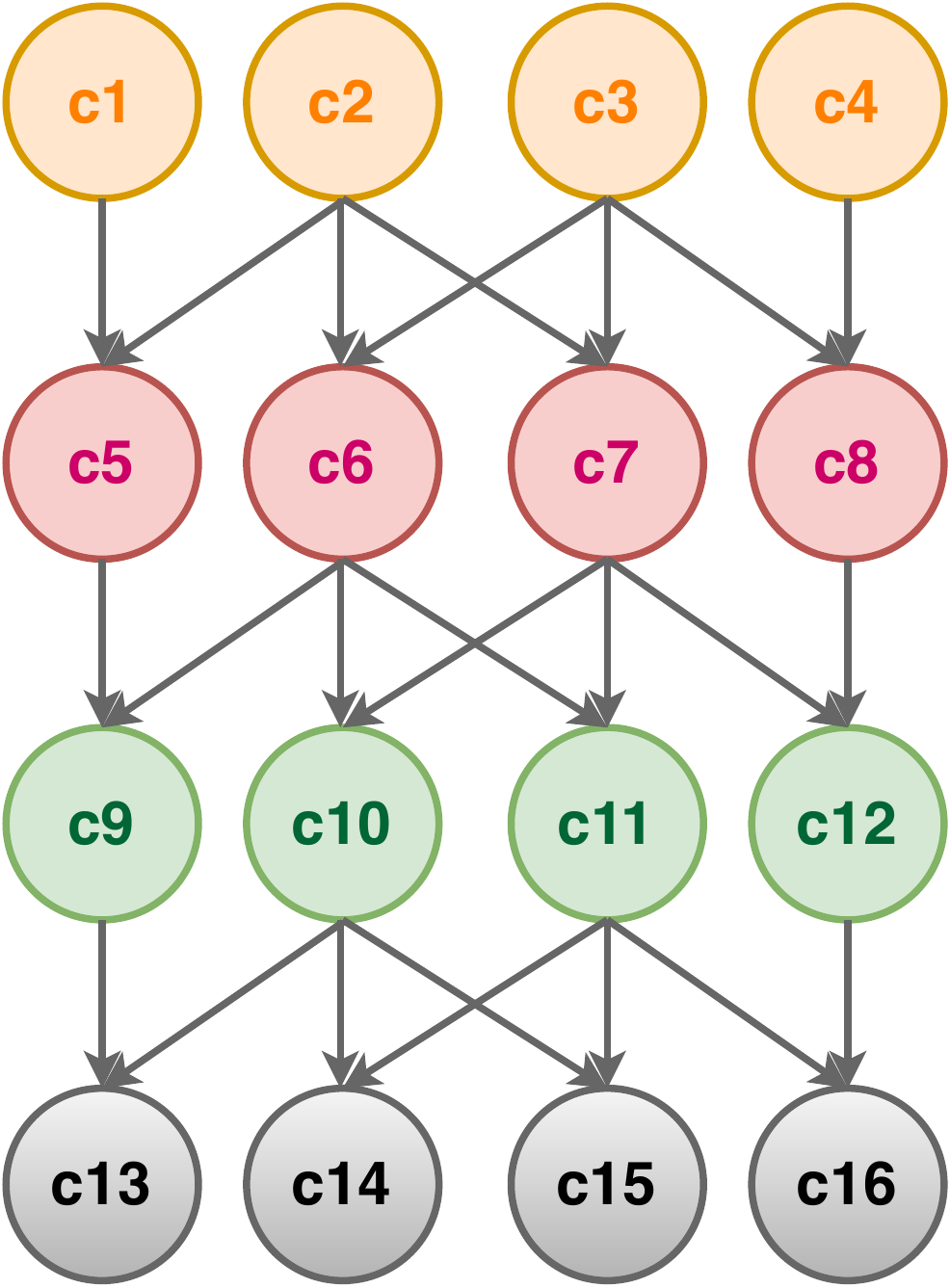}%
	\caption{A DAG job with $16$ coflows.}
	\label{tightDAG}
	\end{subfigure}\hfill
	\begin{subfigure}{0.5\columnwidth}
	\centering
	\includegraphics[width=1.2 in,height=1.6in]{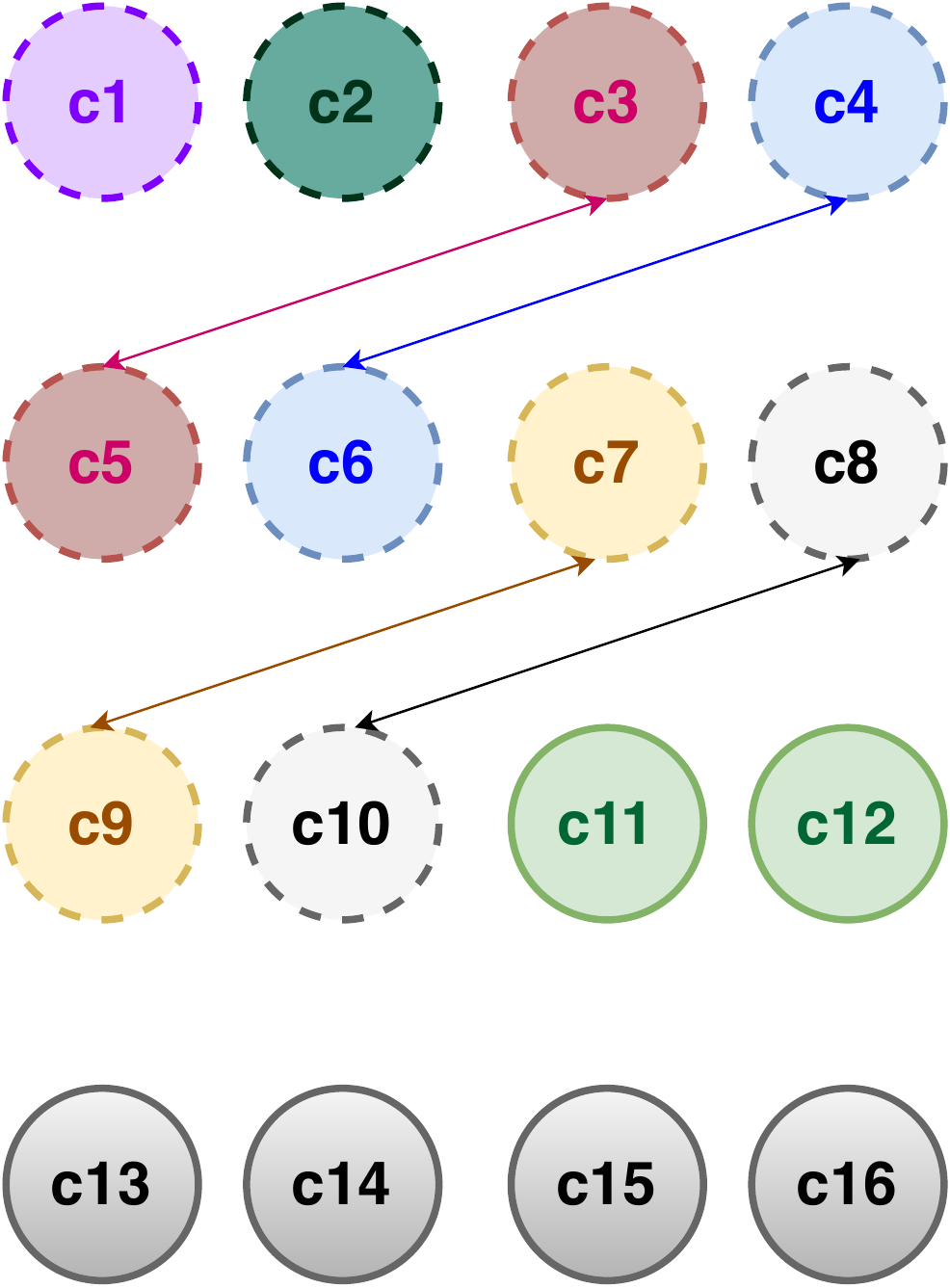}%
	\caption{Scheduling of coflows.}
	\label{tightDAGsche}
	\end{subfigure}\hfill
	\caption{An example of a DAG with $C_{opt}=\Omega(\sqrt{\mu}(\Delta+T))$.}
	\label{tightexample}
\end{figure}

\begin{proof}
Consider a DAG job with $\mu_j$ coflows to be scheduled in an $m \times m$ switch, with $\mu_j=(2K)^2$ for some $K$ and $m > 2K$. Recall that $T_j$ and $\Delta_j$ denote the size of its critical path and its aggregate size, respectively. For simplicity, we drop the subscript $j$. We construct the job as follows. 

First, we describe the demand matrix of each coflow. For coflows $c=1, \dots, 2K$, each coflow has a single flow of size $d$ from server $1$ to server $2$, where $2K=\sqrt{\mu}$ by assumption. These coflows are the root nodes in the job's DAG. For coflows $c=i(2K)+1, \dots, (i+1)(2K)$, $i=1,\dots, (2K)-1$, each coflow has a single flow of size $d$ from server $i+1$ to server $i+2$. 

Now we specify the precedence constraints among coflows. We construct $G_j$ such that its height is $\sqrt{\mu}=2K$ and each of its coflow set has $\sqrt{\mu}=2K$ coflows (see Definition~\ref{CH3def5}).
Consider coflow $c \in S_i$ for $i=1, \dots, H_j-1$. If $i(2K)+1\leq c \leq (i+1/2)(2K)$, then the parent set of coflow $c$ is $\pi_c=\{c^\prime| c-2K \leq c^\prime \leq c-K-1\}$. If $(i+1/2)(2K)+1 \leq c \leq (i+1)(2K)$, then the parent set of coflow $c$ is $\pi_c=\{c^\prime| c-3K+1 \leq c^\prime \leq c-2K\}$. Figure~\ref{tightDAG} shows an example with $\mu=16$. For the constructed DAG, it is easy to see that $T=\Delta=2Kd=\sqrt{\mu}d$.

Next, we specify an optimal schedule for the constructed DAG, and compute its makespan denoted by $C_{opt}$. We first schedule coflows $1, \dots K$, which takes $Kd$ amount of time. We then schedule coflows $K+1$ and $2K+1$ simultaneously. This is feasible since there is no precedence constraint between these two coflows, all the parents of coflow $2K+1$ has been scheduled, and the two coflows do not share a server. Similarly, we schedule coflows $2(i-1/2)K+c$ and $2iK+c$, for $i=1, \dots, 2K-1$ and $c=1, \dots, K$ at the same time. Finally, we schedule the last $K$ coflows,  $c=4K^2-K+1, \dots, 4K^2$ back to back which takes $Kd$ amount of time. For instance, consider the example of Figure~\ref{tightexample}. Coflow $c_1$ and $c_2$ are scheduled back to back from time $0$ to $2d$. Then coflow $c_3$ and $c_5$ get scheduled from $2d$ to $3d$ and so on. Figure~\ref{tightDAGsche} shows the instance at which the first ten coflows (the coflows with dashed lines) are scheduled. The coflows with the same color (that are also linked by an arrow) have been scheduled at the same time.

By scheduling coflows in this fashion, all the precedence constraints and capacity constraints are respected. Moreover, the length of the schedule is $C_{opt}=(2K+1) K \times d=\Omega(\mu d)$. Therefore, $C_{opt}=\Omega(\sqrt{\mu}(\Delta+T))$.
\end{proof}
\section{Proofs of Main Results}
\label{proof}
In this section, we provide detailed proofs of the theorems stating performance guarantees for the proposed algorithms. 
Recall that $g(m)=\log(m)/\log(\log(m))$, and $h(m, \mu)=\log(m\mu)/(\log(\log(m\mu))$.
\subsection{Proofs Related To {\DMA}}
\label{proofDMA}
To prove Theorem~\ref{CH3Thm:mksgenjob} we need the following lemmas. 
\begin{lemma}
	\label{CH3infeaslen}
The length of the infeasible merged schedule (Step 3) is at most $(\mu+1/\beta)\Delta$.
\end{lemma}
\begin{proof}
First note that the isolated schedule for job $j$ in Step 1 spans from $0$ to at most $\mu_j \Delta_j$, since the effective size of each of its coflows is at most $\Delta_j$. By delaying the isolated schedules by at most $\Delta/\beta$,  length of the infeasible merged schedule is at most $\max_j (\mu_j \Delta_j)+\Delta/\beta$ which is bounded from above by $(\mu+1/\beta)\Delta$.
\end{proof}


\begin{lemma}
	\label{CH3numcolli1}
Let $\alpha_t \geq 1$ denote the maximum number of packets that a server needs to send or receive at time slot $t$ in the merged schedule (Step 3). For any $t \in [0,(\mu+1/\beta)\Delta]$, $\mathbb{E}[\alpha_t] =O(g(m))$.
\end{lemma}
\begin{proof}
	 Let $\overline{\calM} :=\calM_S \cup \calM_R$. To prove the lemma, we define random variable $z_{ijt}$ to be 1 if some flow of job $j$ with an end point on server  $i$ is scheduled at time slot $t$. Then $\alpha_t=\max_{i \in \overline{\calM} }\sum_{j \in \calN} z_{ijt}$. Further, note that due to the random delay of jobs' isolated schedules, variables $z_{ijt}$, $j \in \calN$ are mutually  independent. Let $\delta=a g(m)$ for some constant $a$ such that $\delta>1$. Therefore,
	\be
	\label{CH3ineq1}
	\mathbb{E}[\delta^{\alpha_t}] =\mathbb{E}[\delta^{\max_{i \in \overline{\calM} }\sum_{j \in \calN} z_{ijt}}] \leq \expect{\sum_{i \in \overline{\calM} } \delta^{\sum_{j \in \calN} z_{ijt}}}.
	\ee
	Define $p_{ijt}$ to be the probability that $z_{ijt}=1$. By the independent property of $z$ variables, we can write
	\be
	\begin{aligned}
	\label{CH3ineq2}
	\expect{\delta^{\sum_{j \in \calN} z_{ijt}}}&= \Pi_{j \in \calN} \expect{\delta^{z_{ijt}}}\\
	& =\Pi_{j \in \calN} (p_{ijt}\delta+(1-p_{ijt}))\\
	& \leq \Pi_{j \in \calN} e^{p_{ijt}(\delta-1)}=e^{{(\delta-1)}\sum_{j \in \calN}p_{ijt}}\\
	&=e^{(\delta-1)\mathbb{E}[ \sum_{j \in \calN} z_{ijt}]}\leq e^{\beta(\delta-1)},
	\end{aligned}
	\ee
	where the last inequality is due to $\expect{\sum_{j \in \calN} z_{ijt}} \leq \beta$. This is because by choosing delays uniformly at random, $\mathbb{E}[ z_{ijt}]$ is at most the load of job $j$ on server $i$ divided by $\Delta/\beta$, i.e., $\beta d_i^j/\Delta$, where $d_i^j$ is the load of job $j$ (or equivalently the aggregate coflow $\calD^j$) on server $i$ (see Definition~\ref{CH3def1}). Thus, 
	\begin{equation*}
	\expect{\sum_{j \in \calN} z_{ijt}}= \sum_{j \in \calN}\mathbb{E}[ z_{ijt}] \leq \beta, 	
	\end{equation*}
	as $ \sum_{j \in \calN} d_i^j \leq \Delta$ by definition.

	Combining Inequality~\dref{CH3ineq1} and~\dref{CH3ineq2}, and by Jensen's inequality we can write,
	\be
	\label{CH3ineq3}
	\delta^{\mathbb{E}[\alpha_t]} \leq \mathbb{E}[\delta^{\alpha_t}] \leq \sum_{i \in \overline{\calM}} e^{\beta(\delta-1)} = 2me^{\beta(\delta-1)}
	\ee 
	Now, note that if we choose $a$ sufficiently large, then $2me^{\beta(\delta-1)} \leq \delta^\delta$, by definition of $g(m)$. Therefore, we can conclude that $\mathbb{E}[\alpha_t] \leq \delta$, and the proof is complete.
\end{proof}
\begin{lemma}
\label{CH3feaslen}
For any $\epsilon>0$, the probability that the length of the final schedule (Step 4) is greater than $O(g(m))(\mu+1/\beta)\Delta$, is less than $\epsilon$.
\end{lemma}
\begin{proof}
Recall that the constructed merged schedule (Step 3) spans from time $0$ to at most $(\mu+1/\beta)\Delta$ due to Lemma~\ref{CH3infeaslen}. Note that, the length of the final schedule is at most $\sum_{t \in [0,(\mu+1/\beta)\Delta)} \alpha_t$. Using Lemma~\ref{CH3numcolli1} and Markov inequality, for any $\epsilon>0$,
	\begin{equation}
	\begin{aligned}
	\prob {\sum_{t \in [0,(\mu+1/\beta)\Delta)} \alpha_t \geq (a/\epsilon) g(m)(\mu+1/\beta)\Delta }  \leq \epsilon
	\end{aligned}
	\end{equation}
	 Therefore, the proof is complete.

\end{proof}
\begin{lemma}
\label{CH3timecomplexityDMA}
Steps 3 and 4 in {\DMA} can be executed in polynomial time.
\end{lemma}
\begin{proof}
In view of Steps 3 and 4 in {\DMA} algorithm, we may need to run {\BNA} for $(\mu+1/\beta)\Delta$ times. However, in the case that $\Delta$ is not polynomially bounded in $m$, $n$, and $\mu$, we can modify the last step of {\DMA} to ensure that it runs in polynomial time. To do so, define
	 $H=\{\tau_{cj} | c \in G_j, j \in \calN \}  \ \text{and} \  L=\{L_{cj} | c \in G_j, j \in \calN\},$
	 to be the set of all scheduling times and matchings. we sort $H$, and let $\calI$ be the set of time intervals created from elements of $B$, $\calI=\{[h_k,h_{k+1})| k=1,\dots |H|-1\}$. Thus, $\calI$ consists of the time intervals during which the corresponding matching of every coflow is fixed. 
	 
	 For each interval $I$ in $\calI$, we merge the matchings of coflows, namely $L_{cj}(k)$'s, for which the interval $I$ is entirely in the corresponding time interval $[\tau_{cj}(k), \tau_{cj}(k+1))$. In other words, we compute 
	\begin{equation*}
	\calD=\sum_{c,j,k: I \subseteq [\tau_{cj}(k), \tau_{cj}(k+1))} L_{cj}(k). 
	\end{equation*}
Finally, for each merged matching $\calD$, we find an optimal schedule using {\BNA}, i.e., $L,\tau=${\BNA}($l_I \times\calD$), where $l_I$ is length of the interval $I$ of merged matching $\calD$. Then we schedule demand matrix $l_I \times\calD$ according to $L$ and $\tau$.

Note that whenever we run {\BNA}, the number of elements in the list $L$, output of {\BNA}, is at most $m^2$.  This is because according to line 5 in {\BNA}, at each iteration, $t$ is computed such that at least one node becomes tight (i.e., it appears in the set $\Omega$ of line 3 in the next iteration) or a flow completes. Further, $|\tau|=|L|+1$ and the last element of $\tau$ is $D$.
Hence, in view of Steps 3 and 4 in {\DMA} algorithm, we need to run {\BNA} for at most $O(\mu n m^2)$ times as the number of intervals in the set $\calI$ is $O(\mu n m^2)$. Combining this with the fact that {\BNA} runs in polynomial time, the proof is complete.
\end{proof}
We are now ready to prove Theorem~\ref{CH3Thm:mksgenjob}.
\begin{proof}[Proof of Theorem~\ref{CH3Thm:mksgenjob}]
	Steps 1 and 2 in {\DMA} can be executed in polynomial time.  Combining this with Lemma~\ref{CH3timecomplexityDMA}, we can easily conclude that {\DMA} runs in polynomial time.
	
	Moreover, given that $\Delta$ is a lower bound for the optimal makespan, $\beta$ is a constant, and Lemma~\ref{CH3feaslen}, we conclude that makespan of the final schedule is at most $O(\mu g(m))$ of the optimal makespan with high probability.
	\end{proof}

\subsection{Proofs Related To {\DMART} and {\OCDMA}}
\label{proofDMART}
Consider {\DMART}. Let $\alpha_t \geq 1$ denote the maximum number of packets that a server needs to send or receive at time slot $t$ in the infeasible merged schedule (Step 5 in {\DMART}). To prove Theorem~\ref{CH3Thm:mksgenjobRT}, we first state the following lemma that provides a high-probability bound on $\alpha_t$. 

\begin{lemma}
	\label{CH3numcolli2}
For any $\epsilon>0$, $\max_t \alpha_t \leq k_\epsilon \sqrt{\mu_j} h(m,\mu_j)$, with probability greater than $(1-\epsilon)$, for a constant $k_\epsilon$ depending on $\epsilon$, for $t \in [0,\Delta_j/\beta+T_j]$.
\end{lemma}
\begin{proof}
	To prove the lemma, let $P$ denote the probability that any server at any time is assigned more that $\alpha$ packets (to be specified shortly). In what follows we first bound $P_0$ the probability that at least $\alpha$  packets are scheduled to be sent or received by a server $i$ at time $t$. Note that there are at most ${\Delta_j} \choose {\alpha}$ ways to choose $\alpha$ packets from those that have an end point (source or destination) on server $i$. For packet $u$, the probability that it is scheduled at time $t$ is at most $\beta|\calP_{u,j}|/\Delta_j$, where, $\calP_{u,j} \subseteq \calP_{j}$ is the set of path-jobs containing packet $u$ (or equivalently, the coflow to which packet $u$ belongs.). That is because of the random uniform delay for scheduling coflows in $S_0$. More precisely, let $E_{u,t}$ be the event that a specific packet $u$ is scheduled at time $t$ and $P_u$ be the probability that $E_{u,t}$ happens. Furthermore, let $E_{u \in p}$ denote the event that scheduling of $u$ in the final schedule is according to the schedule of path-job $p$.  Then,
	\be
	\label{CH3pu1}
	\begin{aligned}
	P_u & = \sum_{p \in \calP_{u,j}}\mathbb{P}\{E_{u,t} , E_{u \in p}\}\stackrel{(1\star)}=\sum_{p \in \calP_{u,j}} \mathbb{P}\{d_p=t_p,E_{u \in p}\}\\
	& =\sum_{p \in \calP_{u,j}}\mathbb{P}\{E_{u \in p} | d_p=t_p\} \mathbb{P}\{d_p=t_p\}
	\end{aligned}
	\ee
	Equality ($1\star$) is because the probability that packet $u$ is scheduled at $t$ \emph{and} according to the path-job $p$ is equal to the probability that path-job $p$ is delayed by some specific time $t_p$ \emph{and} packet $u$ is scheduled according to the path-job $p$. Regardless of the value of $t$, the probability that path-job $p$ is delayed by $t_p$ is either $\beta/\Delta_j$ or zero (if $t_p<0$). Hence,
	\be
	\label{CH3pu2}
	P_u \leq \frac{\beta}{\Delta_j} \sum_{p \in \calP_{u,j}} \mathbb{P}\{E_{u \in p}| d_p=t_p\} \leq \frac{\beta |\calP_{u,j}|}{\Delta_j}
	\ee
	
	Moreover, for two different packets $u$ and $v$ with at least a common (source or destination) server, the probability that they collide (i.e., are assigned to the same time slot) is zero if they both belong to the same coflow or same path-job, due to the feasible scheduling of each coflow and satisfaction of precedence constraints at each path-job. Otherwise, the probability that the two events $E_{u,t}$ and $E_{v,t}$ happen can be upper-bounded by multiplications of two terms of the form $\beta |\calP_{.,j}|/\Delta_j$ (using arguments similar to Equations~\dref{CH3pu1} and~\dref{CH3pu2}), since the random delays are chosen independently. 
	
	Therefore,
	\be
	\begin{aligned}
	\label{CH3p0bound}
	P_0 & \leq {{\Delta_j} \choose {\alpha}} \Pi_{i=1}^\alpha P_{u_i} \leq (\frac{e \Delta_j}{\alpha})^\alpha  (\frac{ \beta}{\Delta_j})^\alpha  \times \Pi_{i=1}^\alpha |\calP_{u_i,j}| \\
	& \stackrel{(2\star)}\leq (\frac{e\beta\mu_j}{\alpha^2})^\alpha
	\end{aligned}
	\ee
	Note that the size of set $\calP_j$ is bounded by $|S_0|$ (and therefore $\mu_j$) as there is only one path for any coflow in $S_0$ to coflow $R_j$. Therefore, $\sum_{i=1}^\alpha|\calP_{u_i,j}| \leq |\calP_j| \leq \mu_j$. Combining this with the fact that $\Pi_{i=1}^\alpha |\calP_{u_i,j}|$ is maximized when $|\calP_{u_i,j}|=\mu_j/\alpha$, Inequality $(2\star)$ is yielded. 
	
	If we choose $\alpha=k_\epsilon \sqrt{\mu_j} $ then $P_0\leq (m\mu_j)^{-(k_\epsilon-1)}$. Hence, the probability that any server at any time is assigned more that $\alpha$ packets can be bounded by $P \leq 2m (\Delta_j+T_j)P_0 < 2m (\Delta_j+T_j) (m\mu_j)^{-(k_\epsilon-1)}$. This last step is similar to the argument in~\cite{shmoys1994improved, leighton1988universal}, for job shop scheduling problem.
	To specify $k_\epsilon$, note that we require $P$ to be less than $\epsilon >0$, which is satisfied by choosing $k_\epsilon$ as 
	\be
	\label{CH3kepsilon}
&	k_\epsilon \geq \log_{m\mu_j} \big(\frac{2m (\Delta_j+T_j) }{\epsilon}\big)+1
	\ee
	We now need to show that $k_\epsilon$ is a constant by showing that $\Delta_j+T_j$ is polynomially bounded in $m$ and $\mu_j$.
	Let $\delta_j$ denote the maximum size of a flow in job $j$. Note that $\Delta_j+T_j$ is polynomially bounded in $m$, $\mu_j$ and $\delta_j$. In the case that $\delta_j$ is polynomially bounded in $m$ and $\mu_j$, it is easy to see that by choosing $k_\epsilon$ according to~\dref{CH3kepsilon}, with probability $(1-\epsilon)$, there is at most $k_\epsilon(\sqrt{\mu_j}h(m, \mu_j)$ packets on any server at any time.
	If $\delta_j$ is not polynomially bounded in $m$ and $\mu_j$, we round down each flow size $d_{sr}^{cj}$ to the nearest multiple of $\delta_j/m^2\mu_j$ and denote it by $d_{sr}^{\prime cj}$. This ensures that we have at most $m^2\mu_j$ distinct values of modified flow sizes. Therefore, we can treat $d_{sr}^{\prime cj}$ as integers in $\{0,1,\dots,m^2 \mu_j\}$ and trivially retrieve a schedule for $d_{sr}^{\prime cj}$ by rescaling. Let $\calS^\prime$ denote this schedule. If we increase the flow sizes from $d_{sr}^{\prime cj}$ to $d_{sr}^{cj}$ in $\calS^\prime$ by increasing the length of the last matching that flow $(s,r,c,j)$ is scheduled in and achieve schedule $\calS$, we can argue that the length of $\calS$ and $\calS^\prime$ differs in at most $\delta_j$ amount. This is because there are at most $m^2 \mu_j$ number of flows. Thus, length of $\calS$ is at most $(k_\epsilon+1)\sqrt{\mu_j}h(m, \mu_j)$ as $\delta_j \leq T_j$. 
\end{proof}
We are now ready to prove Theorem~\ref{CH3Thm:mksgenjobRT}.
\begin{proof}[Proof of Theorem~\ref{CH3Thm:mksgenjobRT}]
It is easy to see that steps 1-3 of {\DMART} can be done in polynomial time. By Lemma~\ref{CH3timecomplexityDMA}, Steps 4 and 5 of {\DMART} are also executed in polynomial time. Therefore, {\DMART} is a polynomial time algorithm.

Moreover, the completion time of each coflow is bounded by $\Delta_j/\beta+T_j$, since the maximum delay is $\Delta_j/\beta$ and the maximum starting time of coflow $c$ is $T_j-D^{(cj)}$. Using Lemma~\ref{CH3numcolli2}, we conclude that the length of the final schedule is at most $O(\sqrt{\mu_j}h(m, \mu_j))(\Delta_j/\beta+T_j)$ with a high probability. Given that both $\Delta_j$ and $T_j$ are lower bounds for the optimal makespan, the proof is complete.
\end{proof}

We now prove Theorem~\ref{CH3Thm:mksmultigenjobs} regarding performance of {\OCDMA}.
\begin{proof}[Proof of Theorem~\ref{CH3Thm:mksmultigenjobs}]
The proof is similar to proof of Theorem~\ref{CH3Thm:mksgenjob}. Using {\DMART}, completion time of job $j$ is $O(\sqrt{\mu_j}h(m, \mu_j)) \times (\Delta_j/\beta+T_j)$. Delaying and merging these schedules and applying an argument similar to proof of Lemma~\ref{CH3numcolli1} and~\ref{CH3feaslen}, we can conclude that the final solution is bounded from above by $O(\sqrt{\mu}g(m)h(m, \mu)) \times (\Delta/\beta+\max_j T_j)$.
	Combining this with the fact that both $\Delta$ and $\max_j T_j$ are lower bounds on the optimal makespan, we can conclude the result.
\end{proof}

\subsection{Proofs Related to {\OJOD}}
\label{proofGDM}

We use $\tilde{C}_j$ to denote the optimal solution to LP~\dref{CH3RLP} for the completion time of job $j$, and use $\widetilde{\text{OPT}}=\sum_j w_j \tilde{C}_j$ to denote the corresponding objective value. Similarly we use $C_j^\star$ to denote the optimal completion time of job $j$ in the original job scheduling problem, and use $\text{OPT}=\sum_j w_j {C_j^\star}$. The following lemma establishes a relation between $\widetilde{\text{OPT}}$ and $\text{OPT}$. To prove Theorem~\ref{CH3minsumalgper}, we first show the following.

\begin{lemma}
	\label{CH3optlb}
	The optimal value of the LP, $\widetilde{\text{OPT}}$, is a lower bound on the optimal total weighted completion time $\text{OPT}$ of multi-stage coflow scheduling problem, i.e., $\widetilde{\text{OPT}} \leq \text{OPT}$ .
\end{lemma}
\begin{proof}
	It is easy to see that an optimal solution for the original multi-stage job scheduling problem is a feasible solution to LP~\dref{CH3RLP} from which the lemma's statement can be concluded. 
\end{proof}

\begin{lemma}
	\label{CH3bound1}
	If there is an algorithm that generates a feasible job schedule such that for any job $j$, $C_j^{\text{ALG}}=O(\zeta)(T_j+\rho_j+D_j)$, then $\sum_j w_j C_j^{\text{ALG}}=O(\zeta) \times \text{OPT}$, where $C_j^{\text{ALG}}$ is completion time of job $j$ under the algorithm.
\end{lemma}
\begin{proof}
	The proof is similar to the proofs of Lemmas $5$ and $6$, and Theorem $1$ in~\cite{ahmadi2017scheduling}. We first show the following,
	\begin{equation}
	\label{CH3Ineq1}
	\begin{aligned}
	\sum_j w_j C_j^{\text{ALG}} = O(\zeta) (\sum_{i \in \overline{\calM}} \sum_{\calJ \subseteq \calN} \lambda_{i,\calJ} f_i(\calJ)+\sum_{j \in \calN}\eta_j(T_j+\rho_j)),
	\end{aligned}
	\end{equation}
	for $\eta$ and $\lambda$ as computed in Algorithm~\ref{CH3permalg} in Appendix~\ref{app1}. We would like to emphasize that the values of $\eta$ and $\lambda$ at the end of Algorithm~\ref{CH3permalg} constitute a feasible dual solution~\cite{ahmadi2017scheduling}. Note that the second term in the right hand side of inequality~\dref{CH3Ineq1} is the optimization objective in the Dual LP~\dref{CH3DRLP}. Therefore, from weak duality (as $\eta$ and $\lambda$ constitute a feasible dual solution), we can conclude that $\sum_j w_j C_j^{\text{ALG}}= O(\zeta) \times \text{OPT}$.
To show Inequality~\dref{CH3Ineq1}, first note that
\begin{equation}
	\label{CH3wvalue}
	w_j=\eta_{j}+\sum_{i \in \overline{\calM}}\sum_{k \geq j} d_i^j \lambda_{i,k},
\end{equation}
	where, with a slight abuse of notation, $\lambda_{i,k}=\lambda_{i,\calN^\prime}$ when $\calN^\prime=\{1,2,\dots,k\}$. Equation~\dref{CH3wvalue} is correct as job $j$ is added to the permutation in Algorithm~\ref{CH3permalg} only if Constraint~\dref{CH3con1} gets tight for this job. Therefore by the lemma's assumption,
	\begin{equation*}
	\begin{aligned}
	\sum_j w_j C_j^{\text{ALG}}<&O(\zeta) (\sum_j (\eta_{j}+\sum_{i \in \overline{\calM}}\sum_{k \geq j} d_i^j \lambda_{i,k})(T_j+\rho_j+D_j))
	\end{aligned}
	\end{equation*}
	We first bound the term $\sum_j \eta_{j}(T_j+\rho_j+D_j)$. Note that for every job $j$ that has a nonzero $\eta_j$, $T_j+\rho_j > d_{\phi(j)}=D_j$. Therefore,
	\begin{equation}
	\label{CH3help1}
	\sum_j \eta_{j}(T_j+\rho_j+D_j)< 2\sum_j \eta_{j}(T_j+\rho_j).
	\end{equation}
	To bound the term $\sum_j\sum_{i \in \overline{\calM}}\sum_{k \geq j} d_i^j \lambda_{i,k}(T_j+\rho_j+D_j)$, note that for every set $\calN^\prime=\{1,2,\dots,k\}$ with nonzero $\lambda_{i,k}$, we have $T_j+\rho_j \leq d_{\phi(k)}=D_j$. Therefore,
	\begin{equation}
	\label{CH3help2}
	\begin{aligned}
	&\sum_j \sum_{i \in \overline{\calM}}\sum_{k \geq j} d_i^j \lambda_{i,k}(T_j+\rho_j+D_j) \leq 2\sum_j \sum_{i \in \overline{\calM}}\sum_{k \geq j}d_i^j \lambda_{i,k}D_j\\
	&\leq 2\sum_k \sum_{i \in \overline{\calM}}\lambda_{i,k}D_k\sum_{j \leq k}d_i^j \leq 2\sum_k \sum_{i \in \overline{\calM}}\lambda_{i,k} D_k^2\\
	&\stackrel{\star}\leq 4 \sum_{i \in \overline{\calM}} \sum_{\calJ \subseteq \calN} \lambda_{i,\calJ} f_i(\calJ),
	\end{aligned}
	\end{equation}
	where, Inequality ($\star$) is by~\dref{CH3fdef} and the fact that $\lambda_{i,\calJ}$ is only nonzero for the sets of the form $\calJ={1,2,\dots,k}$ for some $k$. Combining~\dref{CH3help1} and~\dref{CH3help2}, Inequality~\dref{CH3ineq1} is derived.
\end{proof}
\begin{proof}[Proof of Theorem~\ref{CH3minsumalgper}]
Recall that $\tilde{C}_j$ is the optimal completion time of job $j$ according to the LP \dref{CH3RLP}. Let $\widehat{C}_j$ denote the actual completion time of job $j$ under {\OJOD}. Also, let $l_j$ be the index of the group to which job $j$ belongs based on \dref{CH3eq:subset}. Let $j_b$ be the last job in group $b$, and $T_b$ be the maximum size of critical paths of jobs in group $b$. Also let $\calT^{(\calJ_b)}$ be the amount of time spent on processing all the jobs in $\calJ_b$. 	 
	Then,
	\be
	\widehat{C}_j&\stackrel{(1\star)}\leq& \max_{k \in b, b \leq l_j} \rho_k+\sum_{b=0}^{l_j} \calT^{(\calJ_b)}\nonumber \\
	&\stackrel{(2\star)} \leq& \max_{k \in b, b \leq l_j} \rho_k+O(\mu g(m))\sum_{b=0}^{l_j} (D_{j_b}+T_b) \nonumber \\
	&\stackrel{(3\star)}\leq & a_{l_j} + O(\mu g(m)) \sum_{b=0}^{l_j}  a_b,
		\label{CH3bound3}
	\ee
	where $g(m)= \log(m)/\log(\log(m))$. Inequality ($1\star$) bounds the completion time of job $j$ with sum of two terms: the first term is the maximum release time of the jobs in the first $l_j$ groups (note that $\max_{k \in b, b \leq l_j} \rho_k$ can possibly be greater than $\rho_j$); The second term is the total time the algorithm spends on scheduling jobs of previous groups plus the time it spends on scheduling $\calJ_{l_j}$.  Lemma ~\ref{CH3feaslen} implies inequality ($2\star$), and inequality ($3\star$) follows from the fact that $\max_{k \in b, b \leq l_j} \rho_k \leq a_{l_j}$, and $D_b$ and $T_b$ are both bounded by $a_b$ for every job $k \in \calJ_b$. From \dref{CH3eq:partition rule},
	\be
	\label{CH3bound4}
	& \sum_{b=0}^{l_j}  a_b = \gamma \sum_{b=0}^{l_j} 2^b = \gamma 2^{(l_j+1)}-1< 2 a_{l_j}.
	\ee
	Combining~\dref{CH3bound3} with~\dref{CH3bound4}, and the fact that $a_{l_{j-1}}=a_{l_j}/2$,
	\begin{equation*}
	\label{CH3bound5}
	\begin{aligned}
	\widehat{C}_j <O(\mu g(m)) a_{l_{j-1}}  \stackrel{(4\star)} <O(\mu g(m)) (T_j+\rho_j+D_j)
	\end{aligned}
	\end{equation*}
	where ($4\star$) is because $T_j+\rho_j+D_j$ falls in $(a_{l_j-1},a_{l_j}]$. This inequality combined with Lemma~\ref{CH3bound1} implies that
	\begin{equation*}
	\begin{aligned}
	\sum_j w_j \widehat{C}_j &\leq O(\mu\log (m)/\log(\log(m))) \times \text{OPT}.
	\end{aligned}
	\end{equation*}
\end{proof}

\section{Conclusion}\label{conclusion}
In this work, we proposed algorithms for scheduling coflows of multi-stage jobs in order to minimize their makespan or total weighted completion time, and provided a performance guarantee for each algorithm. Moreover. simulation results based on real traffic traces showed that our algorithms indeed improve the total jobs' completion times in practice.

This problem is practically well-motivated, involves new challenges, and
deserves further study. As we showed through an example, it is not possible for an approximation algorithm to provide a solution that is within $o(\sqrt{\mu})$ of the two simple lower bounds for a job with a general DAG. An interesting open problem is to improve the results of this paper.

\bibliography{IEEEabrv,references}
\bibliographystyle{IEEEtran}

\appendix[Combinatorial Algorithm]
\label{app1}
In this section, we provide the detailed explanation of the combinatorial algorithm used in {\OJOD} to find a good permutation of jobs and proof of Lemma~\ref{CH3bound1}. 

Recall LP~\dref{CH3RLP}. Define $f_i(\calJ)$ to be the right-hand side of Constraints~\dref{CH3scap} for server $i$ and subset of jobs $\calJ$, i.e.,
\begin{equation}
\label{CH3fdef}
f_i(\calJ)=\frac{1}{2} \big(\sum_{j \in \calJ} (d_i^j)^2 + (\sum_{j \in \calJ} d_i^j)^2 \big).
\end{equation}

 We now formulate dual of  LP~\dref{CH3RLP} as follows:
\begin{subequations}
	\label{CH3DRLP}
	\begin{align}
	\label{CH3DLPobj}
	\max & \sum_{i \in \overline{\calM}} \sum_{\calJ \subseteq \calN} \lambda_{i,\calJ} f_i(\calJ)+\sum_{j \in \calN}\eta_j(T_j+\rho_j) \ \ \mathbf{(Dual~ LP)}\\
	\label{CH3con1}
	& \sum_{i \in \overline{\calM}} \sum_{\calJ: j \in \calJ} d_i^j \lambda_{i,\calJ}+\eta_j \leq  w_j, \ \ j \in \calN\\
	\label{CH3con2}
	& \eta_j \geq 0, \ \  j \in \calN\\
	\label{CH3con3}
	& \lambda_{i,\calJ} \geq 0, \ \ i \in \overline{\calM}, \ \calJ \subseteq \calN. 
	\end{align}
\end{subequations}
The algorithm is presented in Algorithm~\ref{CH3permalg}. Let $\calN^\prime$ be the set of unscheduled jobs, initially $\calN^\prime=\calN$. Also, set $\eta_j=0$ for $j \in \calN$. Define $\Lambda$ to be the set of $\lambda_{i,\calJ}$'s that get specified in the algorithm, and initialize $\Lambda=\varnothing$ (to avoid initializing all the $\lambda_{i,\calJ}=0$, which takes exponential amount of time) (line 1). In any iteration, let $j$ be the unscheduled job with the greatest $T_j+\rho_j$, let $\phi$ be the server with the highest load and let $d_\phi$ be the load on server $\phi$ (lines 3 and 4). Now, if $T_j+\rho_j > d_\phi$, we raise the dual variable $\eta_{j}$ until the corresponding dual constraint is tight and place job $j$ to be the last job in the permutation (lines 5-7). However, if $T_j+\rho_j \leq d_\phi$, we choose job $j^\prime$ as in line 9. Then we define the dual variable $\lambda_{\phi,\calN^\prime}$, set it so that the dual constraint for job $j^\prime$ becomes tight, and place job $j^\prime$ to be the last in the permutation (lines 10-12).
\begin{algorithm}[!h]
	\captionof{algorithm}{Combinatorial Algorithm for Job Ordering}
	\label{CH3permalg}
	Given a set of multi-stage jobs $\calN$: 
	\begin{algorithmic}[1]
	\State $\calN^\prime=\calN$, $\eta_j=0$ for $j \in \calN$, $\Lambda=\varnothing$.
			\For {$k=n,n-1,...,1$}
			\State $\phi(k)=\arg\max_{i \in \overline{\calM}} d_i$
			\State $j=\arg\max_{l \in \calN^\prime} T_l+\rho_l$
			\If {$T_j+\rho_j >  d_{\phi(k)}$}
			\State $\eta_{j}=w_j-\sum_{i \in \overline{\calM}}\sum_{\calJ, j \in \calJ} d_i^j \lambda_{i,\calJ}$.
			\State $\sigma(k)=j$.
			\Else 
			\State $j^\prime=\arg\min_{j \in \calN^\prime}\big( \frac{w_j-\sum_{i \in \overline{\calM}}\sum{\calJ, j \in \calJ} d_i^j \lambda_{i,\calJ}}{d_{\phi(k)}^j} \big)$.
			\State $\lambda_{\phi(k),\calN^\prime}=\big( \frac{w_{j^\prime}-\sum_{i \in \overline{\calM}}\sum{\calJ, j^\prime \in \calJ} d_i^{j^\prime} \lambda_{i,\calJ}}{d_{\phi(k)}^{j^\prime}} \big)$
			\State $\Lambda \leftarrow \Lambda \cup \{\lambda_{\phi(k),\calN^\prime}\}$.
			\State $\sigma(k)=j^\prime$.
			\EndIf
			\State $\calN^\prime \leftarrow \calN^\prime / \sigma(k)$.
			\State $d_i \leftarrow d_i-d_i^{\sigma(k)}, \ \forall i \in \overline{\calM}$.
			\EndFor
			\State Output permutation $\sigma$.
		\end{algorithmic}
		
\end{algorithm}

\end{document}